\documentclass[12pt]{amsart}
\usepackage[utf8]{inputenc}
\usepackage{natbib}
\usepackage{graphicx} 
\usepackage{epstopdf}
\usepackage{comment}
\usepackage{lscape}
\usepackage{xcolor} 
\usepackage{pgfplots}

\usepackage{placeins}

\newcommand{\figmode}{0} 

\usepackage[left=1.6in,top=1.8in,right=1.6in,bottom=1.8in]{geometry}
\usepackage[onehalfspacing]{setspace}
\usepackage[foot]{amsaddr}

\def\I{\mathcal{I}}
\def\S{\mathcal{S}}
\def\C{\mathcal{C}}

\def\ep{\varepsilon}
\def\one{\mathbf{1}}
\def\then{\Longrightarrow}

\def\argmax{\mbox{argmax}}

\def\la{\lambda}

\def\Re{\mathbf{R}}

\def\os{\emptyset}

\newcommand{\df}[1]{\textit{#1}}

\newcommand{\abs}[1]{ \left | #1 \right | }

\theoremstyle{plain}
\newtheorem{theorem}{Theorem}
\newtheorem{lemma}[theorem]{Lemma}

\theoremstyle{remark}
\newtheorem{remark}[theorem]{Remark}

\title[Inferring Rankings from Scores and Selections]{Prestige in Numbers: How Test Scores and Choices Reveal School Rankings}
\author{Federico Echenique and Michael Olabisi}
\date{\today}
\thanks{The authors have relied upon data supplied by the Graduate Management Admission Council (“GMAC”) to conduct the independent research that forms the basis for the findings and conclusions stated by the authors in this article. These findings and conclusions are the opinion of the authors only and do not necessarily reflect the opinion of GMAC}

\begin{document}
\begin{abstract}
This paper introduces a novel revealed-preference approach to ranking colleges and professional schools based on applicants' choices and standardized test scores. Unlike traditional rankings that rely on data supplied by institutions or expert opinions, our methodology leverages the decentralized beliefs of potential students, as revealed through their application decisions. We develop a theoretical model where students with higher test scores apply to more selective institutions, allowing us to establish a clear relationship between test score distributions and school prestige. Using comprehensive data from over 490,000 GMAT test-takers applying to U.S. full-time MBA programs, we implement two ranking methods: one based on monotone functions of test scores across schools, and another using score-adjusted tournaments between school pairs. Our approach has distinct advantages over traditional rankings: it reflects the collective judgment of the entire applicant pool rather than a small group of experts, and it utilizes data from an independent testing organization, making it resistant to manipulation by institutions. The resulting rankings correlate strongly with leading published MBA rankings ($\rho = 0.72$) while offering the additional benefit of being customizable for different student subgroups. This method provides a transparent alternative to existing ranking systems that have been subject to well-documented manipulation.
\end{abstract}
\maketitle

\section{Introduction}\label{sec:intro}  \def\GMAT{ GMAT\texttt{\texttrademark} }

Colleges and professional schools draw hundreds of thousands of  students every year on the basis of rankings that can be subject to manipulation. Getting fair estimates of the value that schools represent to prospective students before their application and enrollment decisions is a multifaceted challenge, one that has inspired a rich literature on the value of education \citep[e.g.,][]{angrist2023methods,mountjoy2021returns,macleod2017big}, on how students prepare for admissions to high-reputation schools \citep[e.g.,][]{macleod2015reputation}, and on the incentives of schools to rig the rankings \citep{dearden2019strategic,averyetalRPranking,ehrenberg2003method}, as well as notable articles in the popular press \citep[e.g.,][]{tough2023americans,columbia}.

We propose a new \textit{revealed preference approach} to ranking colleges and professional schools. Our approach is based on where students choose to apply, and on the test scores that the students receive in a standardized test. Our methodology builds on a portfolio-choice model meant to capture where students apply. Under certain assumptions, we can derive a comparative statics result which says that students with higher standardized test scores are more aggressive ---they apply to ``better,'' more selective, schools--- than students with lower standardized test scores. 

The comparative statics provide us with a reduced-form connection between the score distribution of the applicants to a school and the prestige, quality, or selectivity of the school in question. Our theoretical results say that if one had access to the relevant ranking of schools, then one would observe that if school A is ranked higher than B, then the test score distribution of the students who apply to A would be better than the distribution of the students who apply to B. To operationalize our method, we propose ways of backing out a ranking for which this property holds true.

We apply our methodology to the ranking of U.S.\ full-time MBA programs and the \GMAT exam. After completing the exam, students must indicate to the test administrator which schools should receive their test scores as an indication of the schools where they have submitted, or intend to submit applications. We implement two versions of our methodology: one backs out a ranking by imposing that a certain class of functions is monotone increasing as colleges improve in the ranking. The second version is a score-adjusted tournament between school pairs. With the model, we are able to back out a common order, or ranking, over colleges purely with the information about where students apply and how they performed on the standardized test.



College and school rankings are common and influential. They rest on different methodologies, from data supplied by the colleges to surveys of prospective students.  Our approach is distinctive in a number of ways:

First, our ranking reflects the decentralized beliefs over the universe of schools held by the population of potential students. The ranking we obtain is an  estimate of a common ranking that would underlie the students' views about the schools. In particular, and in contrast with some of the best-known rankings in the popular press, it does not reflect the views of a small coterie of experts or insiders. It reflects the broad views of the whole population of potential students.

Second, our ranking uses data on standardized tests that is available from the test administrator. As a consequence, the source of the data is an independent body that is not under the influence of any particular college. The rankings of, for example, US News \& World Report (USNWR) require data that must be supplied by college administrators. This is problematic. There are documented instances of selective reporting and data manipulation. A Columbia University mathematics professor made very specific accusations against his university of rigging the process that gave it the USNWR No.\ 2 ranking in 2022 \citep{columbia}. The claims of dishonesty include fudging class sizes and the percent of full-time faculty with terminal degrees. Temple University's MBA program had an equally notable scandal. Its former business school dean was convicted in court in 2021 of using fraudulent data to boost school rankings before his removal \citep{temple}. These scandals exist side-by-side with unease in academia about the echo-chamber problem of USNWR and similar ranking systems, which use school administrators' ratings of other schools for a notable share of the ranking weight. Students (and potential students), the primary audience for the rankings, and whose tuition dollars are at stake, have no voice in these measures of reputation. Taking our approach, which organizes the features and choices of the vast pool of potential students, steps around the well-documented incentives of schools to rig the rankings  \citep[c.f.,][Section VI.]{averyetalRPranking}.

Ours is not the first paper with a revealed-preference approach to ranking colleges. \cite{averyetalRPranking} use survey data on where \emph{admitted} students decide to enroll after being offered admission. They use a discrete-choice model to estimate a common-value component that seeks to capture a measure of school quality. Their results are based on a very selective survey: In the authors' words, their survey captures ``students with very high college aptitude who are
likely to gain admission to the colleges with a national or broad regional draw.'' For example, the average SAT score in the authors' sample is in the 90th percentile of the SAT score distribution, and 47\% of the students in the sample applied to at least one Ivy League school. As a consequence, Avery et al.\ are only able to rank a reduced set of schools and the ranking is not representative of the whole population of students. Data on admitted students' choices for all colleges or business schools are generally unavailable to any individual or institution. In contrast, we are able to work with a comprehensive data set: the choices of \emph{all} US-based candidates applying to full-time US MBA programs, high-scoring and low-scoring alike. This encompasses more than 1.7 million score reports that show the schools and programs selected by more than 490,000 unique candidates over the nine-year period in our data.

A crucial methodological difference with \cite{averyetalRPranking} is that they consider students' enrollment choices from a ``budget'' of schools at which the student was accepted. Our approach relies instead on the set of colleges selected by candidates --- that is, the set of schools picked to receive candidates' \GMAT scores. The inference from enrollment used in \cite{averyetalRPranking} seems more transparent and easier to connect to preferences, except that the budget of schools is really endogenous. So ratings from admitted candidates' choices may be tainted by endogeneity. Students choose where to apply based on their preferences and characteristics, and the budget that they then choose from depends on where they are admitted. For example, imagine that there are two kinds of students: those who want a technical career and those interested in a liberal arts education. The students might either apply mostly to liberal arts colleges, say to A and C; or mostly to engineering schools, B and D. Now imagine that most students would rank A and B as the two top schools, but the engineering students who apply to B and D tend to prefer D if they get admitted. Then it is possible for the Avery et al.\ methodology to conclude that D should be ranked at the top because it tends to be chosen over B by the students who are given such a choice, even if most students actually rank B over D (this example is based on the discussion in \cite{averyetalRPranking} explaining Caltech's place in their ranking).

We have emphasized the strengths in our approach, but the work of \cite{averyetalRPranking} has advantages as well. By considering the actual choices made by the students, they do not need to rely on the connection between ranking and test scores that we establish theoretically. In consequence, their methodology does not rely on the assumptions we make in our theory (they make other assumptions, though). Moreover, the survey data in their paper contains a richer set of covariates and allows them to consider empirical questions that are not possible with our administrative data.

Our theoretical model captures students' selection of an application portfolio. The problem of choosing an \emph{application portfolio} of schools is a unique kind of optimization problem. A student chooses a set of schools but has in mind an ordering of the schools based on preferences. This means that the third-best school that the student applies to, for example, is only relevant in the event that the student is rejected from their first- and second-best choices. The portfolio selection problem presents unique challenges and was first studied by \citet{chade2006simultaneous}. The work of Chade and Smith, when applied to the college application problem, assumes that colleges make independent admissions decisions. The more recent paper by  \citet{alishorrer24} assumes that admissions decisions are based on common test scores, and is therefore much closer to our theoretical results. In fact, while the models in the two papers are somewhat different, Proposition~1 in \cite{alishorrer24} is closely related to our main theoretical findings. Their comparative statics results say that bad news about test performance translates into a less aggressive application portfolio; this is similar to one of the ideas that go into our Theorem~\ref{thm:main}. Our theoretical results also connect the score distribution to the underlying ranking of colleges --- there is no analogous connection in any of the preceding theoretical literature. Our paper also differs from the theoretical literature in purpose and motivation: \citet{chade2006simultaneous} is a general study of optimal portfolio problems (there is a follow-up paper that embeds the problem in a market-level equilibrium model of college applications: see \citet{chade2014student}). The model in \citeauthor{alishorrer24} is based on subjective beliefs over test performance, rather than on applicants' selection of portfolios after taking the test. They then try to explain correlations in college decisions, while we seek to infer a formal measure of the unobserved common value of each school from test candidates' observable application decisions.

Our main empirical results are rankings of full-time MBA programs (FTMBA) based on student selections. Our rankings have a correlation coefficient of 0.72 with the leading published FTMBA ranking for US universities in 2018, but they are, by construction, insulated from manipulation by universities and business schools. The rankings have one other notable advantage: they can be tailored to specific sub-groups of candidates (much like \cite{averyetalRPranking}, but without needing a specialized survey to gather preferences). For example, how the FTMBA helps the career of a business or economics major may differ from how it helps an engineering major who wishes to transition into a consulting or management role. Other subgroups with systematically different needs can have rankings that reflect the choices of others in their situation. 

The next section describes the model. Section \ref{sec:meth} explains our estimation strategy and describes the data. Section \ref{sec:res} presents the estimation results, followed by a discussion of the features of the model in Section \ref{sec:discussion} and the conclusion. 

\section{Model}\label{sec:model}
\subsection{Model Setup}\label{sec:setup}
We propose a model in which potential students choose an optimal portfolio of schools to apply to, while schools are selective in their admissions decisions. Colleges are expected to offer admission only if the sum of the features in each application package and a stochastic component exceeds the institution's admission threshold. Higher-ranked schools have higher requirement thresholds for admitting a student than do lower-ranked schools. Our model yields a method for estimating this ranking based on the idea that the most prestigious universities are also the most selective on the observable features of candidates. 

The primitive elements of our model are:
\begin{itemize}
    \item A set $\C=\{c_1,\ldots,c_m\}$ of all \df{colleges}, and for each  college $c_j\in \C$, a \df{threshold} $t_j$.
        \item A set $\S= \{s^1,\ldots, s^T \} \subseteq \Re$ of all possible \df{test scores}, with $s^1> s^2 > \cdots > s^T$. For example, GMAT exam scores go from 200 to 800 in 10 point increments.
    \item A set $\I = \{1,\ldots, n\}$ of all \df{students}, and for all $i\in \I$, a score $s_i\in \S$. 
    \item An absolutely continuous distribution function $F$.
    \item A probability measure $\la$ on $\Re^{\C}$.
\end{itemize}

\emph{Colleges}: Given is a set of colleges $\C$. Each individual college is identified with a particular program. When two programs are offered by the same institution, we treat them as two different colleges. 

\emph{Students}: Each student $i\in \I$ has a utility function $v_i: \C\to \Re_{++}$ over colleges. We write the function as $v_i(c_j)=v_{i,j}$, for $j=1,\ldots,m$.  

For technical convenience, we assume that each student regards each school $c$ as a collection of identical programs: effectively that there are ``duplicate'' schools. So that, if the student wanted, they could apply multiple times to the same school. As we explain in detail in Section~\ref{sec:discussion}, the assumption of duplicate colleges is minor and done for technical convenience.

\emph{Tests and admissions}:  Each student $i$ receives a test score $s_i\in \S$. A college $c_j$ evaluates student $i$ according to their test scores, as well as an idiosyncratic component $\ep_{i,j}$ that is drawn in iid fashion from the distribution $F$. 

Concretely, college $c_j$ has an admissions \df{threshold} $t_j$. If $i$ applies to $c_j$, they are admitted if \[ 
s_i + \ep_{i,j} \geq t_j.
\] In consequence, the probability that student $i$ is accepted at college $c_j$, if they apply there, is 
\[ 
p_i(c_j) = F(s_i - t_j)
\] 

The set of thresholds defines a selectiveness ordering $\leq$ on colleges.  We say that $c_j \leq c_h$ if $t_j\leq t_h$. So college $c_j$ comes before college $c_h$ in the selectiveness ordering if $c_j$ uses a less selective threshold than $c_h$. Thus if $c_j\leq c_h$ then $p_i(c_h)\leq p_i(c_j)$ for all students $i$. 

Test scores are particularly subject to selectivity, given that the average scores of accepted students feature in the metrics that expert opinion sources use to rank schools (e.g. US News and World Report). Schools are likely to `game' the ranking system by selecting candidates with high test scores, even if those candidates are not as strong on other observable characteristics \citep[e.g.,][]{averyetalRPranking,dearden2019strategic}. 

\emph{Portfolio selection problem}: Students are required to send their test scores to schools as part of the test-taking process. Schools receive the test scores directly from the test administrator (and students are responsible for submitting other application materials to schools). Test scores in our model measure school readiness, not necessarily innate ability \citep[c.f,][]{fu2014equilibrium}.

Student $i\in \I$ with utility function $v_i$ chooses a \df{portfolio} $A\subseteq \C$ of colleges to apply to. The expected utility for $i$ of applying to this set of colleges can be calculated as follows. 

Enumerate $A$ as  $\{ c^1,\ldots, c^K \} $ with $v_i(c^1)\leq v_i(c^2) \cdots \leq v_i(c^K)$. Then, 
\begin{align}
U_i(A) = p_{iK}v_{iK} + p_{iK-1}v_{iK-1}(1-p_{iK}) +\ldots + p_{i1}v_{i1}\prod_{l=2}^K(1-p_{il}), \label{eq:pfolio}
\end{align}
where 
$p_{ik} = p_s(c_k)$ is the probability of college $k$ accepting the candidate $i$ with the score $s$, and $v_{ik} = v_i(c_k)$.

The portfolio selection is driven by a candidate-specific ranking of schools, given the expected benefits that admission to each school confers. Crucially, the portfolio selection is different from other problems of optimal choice among uncertain prospects. The student will only attend one college, so the expected utility of a portfolio is driven by the best college that the student gets admitted to. Such portfolio choice problem have been studied by \cite{chade2006simultaneous} and \cite{alishorrer24} (see our discussion of these papers in the introduction).

Let $C_s(v)\in \argmax\{U_i(A):A\subseteq \C, \abs{A}\leq k^v_s\}$ be the set of colleges chosen by a student with score $s$ and utility function $v:\C\to \Re$. The student with score $s$ and utility function $v$ is constrained to $k^v_s$ applications.



\emph{Extension:} The model can be extended to capture application requirements that go beyond test scores and purely idiosyncratic considerations. For example, the work of \cite{averyetalRPranking} allows for covariates and individual unobserved heterogeneity, as long as it affects the agents in an additive fashion. Our model allows for similar, additively separable, extensions of the model. Let $o_i$ denote ``other features'' that pertain to a candidate $o_i$; for example, extracurricular activities or real-world experiences. Let $t_j+o_j$ be the modified threshold used by a college that takes such other features into account. The threshold condition for admission then becomes:
\begin{align}
s_i + o_i + \ep_{i,j} > o_j + t_j, \label{eq:thresh}
\end{align}

Highly ranked schools are therefore more likely to admit high-scoring students, as well as students with higher levels of other observable features required by schools. The model then reduces to ours if we make some additional assumptions: Specifically if  $o_i = \delta s_i + \zeta_{ij}$ and $o_j = \delta t_j + \zeta'_{ij}$, with $\zeta_{ij}$ and $\zeta'_{ij}$ being purely idiosyncratic iid terms. Then equation \eqref{eq:thresh} becomes $ E\left[(1 + \delta) s_i + \ep_{i,j}\right] > E\left[(1 + \delta) t_j\right]$, so that in the aggregate, one still expects $F'[(1 + \delta)(s_i - t_j)] = F(s_i - t_j)$. The probability of admission remains linked to test scores.

\subsection{Main Theorem}\label{sec:mainthm}

Consider the following (upper tail of) the distribution of schools:
\[ 
\bar G_s(c) = \Pr ((\exists c'\in \C)(c'\geq c \text{ and } c' \text{ is chosen by a student with score } s)).
\] So that $1-\bar G_s(c)$ is the distribution of schools that a student with score $s$ applies to. For notational convenience, we work here with the upper tail instead of the usual cumulative distribution function.

\begin{theorem}\label{thm:main} Suppose that $F$ is concave and that $s\leq s'$ implies $k^v_s\leq k^v_{s'}$.
For any $c\in \C$ and $s,s'\in \S$, $s\leq s'$ implies that $\bar G_{s'}(c)\geq \bar G_s(c)$.
\end{theorem}

Theorem~\ref{thm:main} says that the distributions of college applications for each score are increasing in the score in the sense of first-order stochastic dominance. The proof of the theorem is in Section~\ref{sec:proof}.

Importantly, the distribution $\bar G_s(c)$ is unobservable. It relies on the ordering $\geq$, which is precisely the object that we seek to estimate from the data: the relevant ranking over schools. What is observable is the portfolio of schools that each student applies to. We leverage Theorem~\ref{thm:main} to, very loosely speaking, invert $\bar G_s(c)$ and thus recover the ranking $\geq$.

\section{Methods and Data}\label{sec:meth}
\subsection{Methods}\label{sec:m}

The core of our methodology is that students with high test scores tend to apply to highly-ranked colleges. It works because it aligns the desirability of a college with its selectiveness. The colleges with greater prestige, or perceived quality, attract more applications. As a consequence, the probability of acceptance for these colleges is low. Selectivity creates an advantage for candidates with higher test scores, who will send their scores to the selective programs, while low-scoring candidates will be less likely to do so.

While intuitive, this core idea is subtle.  High-scoring students have a higher probability of being accepted {\em anywhere} than students with low test scores. It is not at all obvious that they will favor the better colleges. In fact, a high-scoring agent might apply to a subset of very selective schools and then offset this by including a larger number of safety schools, or perhaps some schools that are safer than they would have applied to if their test scores had been lower. Theorem~\ref{thm:main} establishes that, under the assumptions in our model, the comparative statics of applications in test scores work as described. The ideas and assumptions behind our result are discussed further in Section~\ref{sec:discussion}.

\subsection{Empirical Strategy}\label{sec:empiricalstrategynew}
We consider two approaches for inferring the aggregate ranking of preferences from student selections.

The challenge is to back out a ranking for which the statement in Theorem~\ref{thm:main} is valid. We need an ordering $\geq$ for which, for any $c\in \C$ and $s,s'\in \S$, $s\leq s'$ implies that $\bar G_{s'}(c)\geq \bar G_s(c)$. Naively, one might try to search for a ranking that delivers the property in our theorem, but this approach is not feasible computationally. The set of possible rankings is extremely large and has no structure that allows for a ``gradient style'' algorithm.\footnote{For example, with $100$ schools (and we have many more in our data) the number of possible rankings is approximately $9.3\times 10^{157}$; a number that far exceeds the number of particles in the universe.} Instead, we leverage some simple implications of the theorem.  

\subsubsection{The $m$ measure and monotonicity in rank.}\label{sec:i1}

To operationalize Theorem~\ref{thm:main}, we consider a class of functions that are meant to be monotone in rank. This allows us to recursively identify the top-ranked schools, out of those that have not been ranked yet. The method discussed in this section is practical and serves as the basis for the results that we report in Section~\ref{sec:res}. Appendix~\ref{sec:altempiricalstrategy} presents an alternative methodology that is potentially useful with a different dataset than ours.

First, we identify the top school as the one that is ``most monotone increasing'' with score.  Specifically, let $g_s(c)$ be the fraction of students with score $s$ that apply to school $c$. Consider the quantity \[
m_1(c) = \sum_{s,s'\in S}(g_{s'}(c) - g_s(c))(s'-s).
\] Note that if $c$ is the top school, then $s\mapsto g_s(c)$ is monotone increasing, so that $(g_{s'}(c) - g_s(c))(s'-s)>0$.

Suppose that $c_1$ has been identified as the top school. Now consider each remaining school $c$, then if $c$ were number two then for each score the right tail of the distribution is \[ 
\bar G_s(c) = g_s(c)+g_s(c_1).
\] Now, one may consider 
\[
m_2(c) = \sum_{s,s'\in S}(\bar G_{s'}(c) - \bar G_s(c))(s'-s)\] \[= 
\sum_{s,s'\in S}(g_{s'}(c)- g_s(c))(s'-s) + \sum_{s,s'\in S}(g_{s'}(c_1) -g_s(c_1) )(s'-s)
\]
and choose the $c\neq c_1$ that maximizes this expression. Given that only the first term,
$\sum_{s,s'\in S}(g_{s'}(c)- g_s(c))(s'-s)$, depends on $c\neq c_1$, this school is the one yielding the second-highest value of $m_1(c)$.

Therefore, the ranking produced by $m_1$ yields an estimate of the ranking of schools that is captured by $\geq$.

An alternative approach is based on another class of functions that are monotone increasing in the ordering of schools. As a first step, we may find the top ranked school $c_1$ by maximizing \[
m^+(c) = \sum_{s,s'\in S}1_{(g_{s'}(c) - g_s(c))(s'-s)>0}.
\] This function simply counts the number of times that we have the ``correct'' correlation between scores and school: $(g_{s'}(c) - g_s(c))(s'-s)>0$. It stands to reason that $c_1$ maximizes this function.

Then we may proceed recursively. Suppose that we have determined that  $c_1>c_2>\cdots c_{k-1}$. Now we may find $c_K$ by letting $\bar G_s(c) = g_s(c) + \sum_{k=1}^{K-1}g_s(c_k)$ and maximizing
\[
m^+(c) = \sum_{s,s'\in S}1_{(\bar G_{s'}(c) - \bar G_s(c))(s'-s)>0}.
\]
If indeed the top colleges are $c_1>c_2>\cdots c_{k-1}>c$ then $\bar G_s(c)$ would be monotone increasing in $s$ and hence $(\bar G_{s'}(c) - \bar G_s(c))(s'-s)>0$.

\subsubsection{A tournament}\label{sec:i2}

The ideas in Theorem~\ref{thm:main} can be operationalized through a tournament. Consider two schools $c$ and $c'$, and two students that have the same utility function over schools. Suppose that the first student has a higher test score, and that the first student applies to $c'$ while the second applies to $c$. (Each student chooses one school but not the other). Then we might infer (in line with Theorem~\ref{thm:main}) that $c' > c$. This idea leads to a tournament between the colleges: the presence of such a pair of students is an indication that $c'$ ``scores a point'' over $c$ in the construction of a ranking of schools. Our results imply that such scoring will lead to schools that are highly ranked according to $\geq$, beating schools that are ranked below. 

Indeed, for a fixed utility function over schools,  Lemma~\ref{lem:paloapalo} in the proof of Theorem~\ref{thm:main} means that the scoring method we have just outlined is guaranteed to compare more pairs of schools correctly (meaning, rank them like $\geq$ does) than incorrectly (in violation of $\geq$). For example, a conservative bound for a student that applies to 5 schools translates into a guaranteed 33\% more correct than incorrect comparisons. This number is a theoretical worst-case guarantee: In real instances, we expect there to be many more correct than incorrect comparisons.

To formalize our tournament method, 
let $R(c',c)$ denote the set of students $i$ for which there is some student $j$ with $s_j<s_i$ so that $j$ applies to $c$ while $i$ applies to $c'$. This is in some sense the set of students who reveal $c'$ to be ranked above $c$.

Now suppose that $c'$ is ranked above $c$ if the number of students in $R(c',c)$ exceeds the number of students in $R(c,c')$. Effectively, we obtain a tournament graph between schools: a binary relation between schools describing when one school beats another according to the scoring method we have described. Now we may consider, for each school $c$, the number of schools that $c$ \emph{beats} in the tournament. Such a number gives us a ranking of all the schools.

\subsection{Data}\label{sec:datadesc}
We use a proprietary dataset with all GMAT tests taken in the U.S. between July 2006 and June 2015. The anonymized data from the Graduate Management Admission Council (GMAC), which administers the GMAT exam, comes as a set of \emph{score reports}. Each score report links a test taken by an anonymized test candidate, on a given date, to the uniquely identified school programs selected by the candidate. In the period covered by our data, candidates took the GMAT exam as a requirement for applying to any accredited full-time MBA (FTMBA) program. On completing the test, candidates select the school programs to receive their score reports - they can select up to five school programs at no additional cost.\footnote{Reports for each program selected after the test-day, or for programs in exceed of the five, cost an additional score reporting fee -- at that time \$28. Selections are not necessarily applications to a school, but are required by FTMBA programs as part of applications.} 

Our analysis draws on more than 1.7 million score reports that show how 490,000 prospective MBA students selected one or more schools out of a menu of 688 full-time MBA programs in the US. We observe the GMAT \emph{Total Score} for each candidate. (Candidates are represented with a unique record ID to preserve privacy). 
Test scores range from 200 to 800. (This is because we use observations from the generation of the GMAT exam which was discontinued on January 31, 2024. The current GMAT Focus Edition has a different score scale and distribution). The full database holds more than 3 million GMAT score reports, with scores from about 1.2 million GMAT tests completed in the nine July-to-June cycles (or test-years) that we study.

We focus on the score reports from US-based test-takers sent to full-time MBA programs (FTMBA) in the U.S., (and drop the balance of score reports that went to non-MBA master's programs, part-time MBA programs, Ph.D.\ programs, or foreign programs). The strong correlation between score-sending and school applications (with an $R^2$ of about 0.98 according to GMAC), gives us confidence in using where students select as an approximation to where they apply. The focus on U.S.\ FTMBA programs is motivated by their prevalence in the data, and the fact that as a relatively homogeneous group with a uniform requirement for the GMAT exam, they provide the clearest \emph{apples-to-apples} comparison for our goal: using candidates' program selections to infer aggregate school preferences. As a data cleaning exercise, we also drop FTMBA programs that received fewer than 122 score reports.\footnote{Dropping programs with fewer than 122 (or 61*2) score reports, avoids the noise created in estimates of $m$-measures by schools that could not possibly receive enough score reports from candidates to cover the range of 61 possible test scores.}

Some features of the data are notable, and affirm our approach. The first, shown in the top panel of Figure \ref{fig:scoredist2}, suggests that the test score distribution resembles a skewed normal distribution, with bunching at the test scores that appear to be key school admissions thresholds. The pattern of bunching is consistent with our rationale that test scores are a crucial part of the signal that candidates send to colleges and business schools. The top panel of Figure \ref{fig:scoredist2} also affirms this paper's claim that candidates with higher scores reach more for the most selective programs, as illustrated by the differences in the skew of the test score distributions linked to a high-reputation school like Harvard, versus a token sample of lower-ranked schools.

\ifnum\figmode=0
\begin{figure}[h]\caption{Grouped Test Score Distribution}
\label{fig:scoredist2}
\begin{center}
\includegraphics[width=0.95\columnwidth]{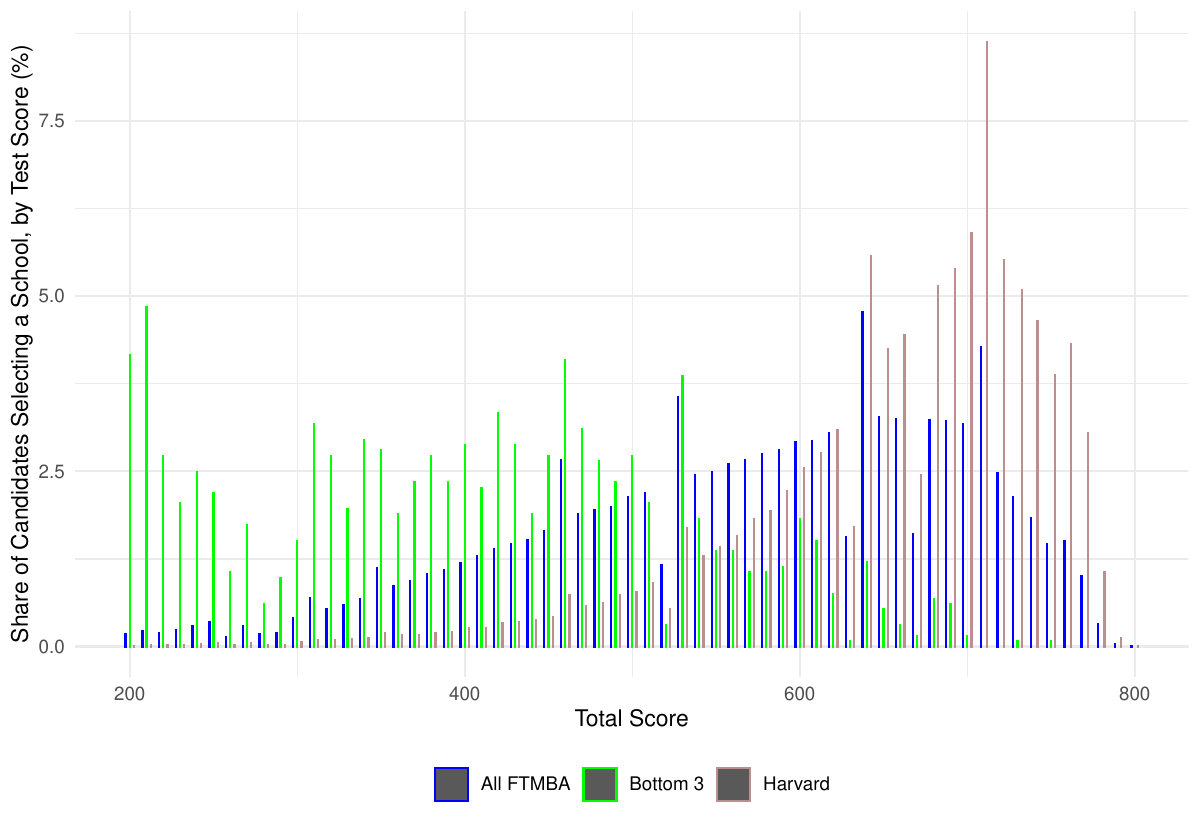}
\includegraphics[width=0.95\columnwidth]{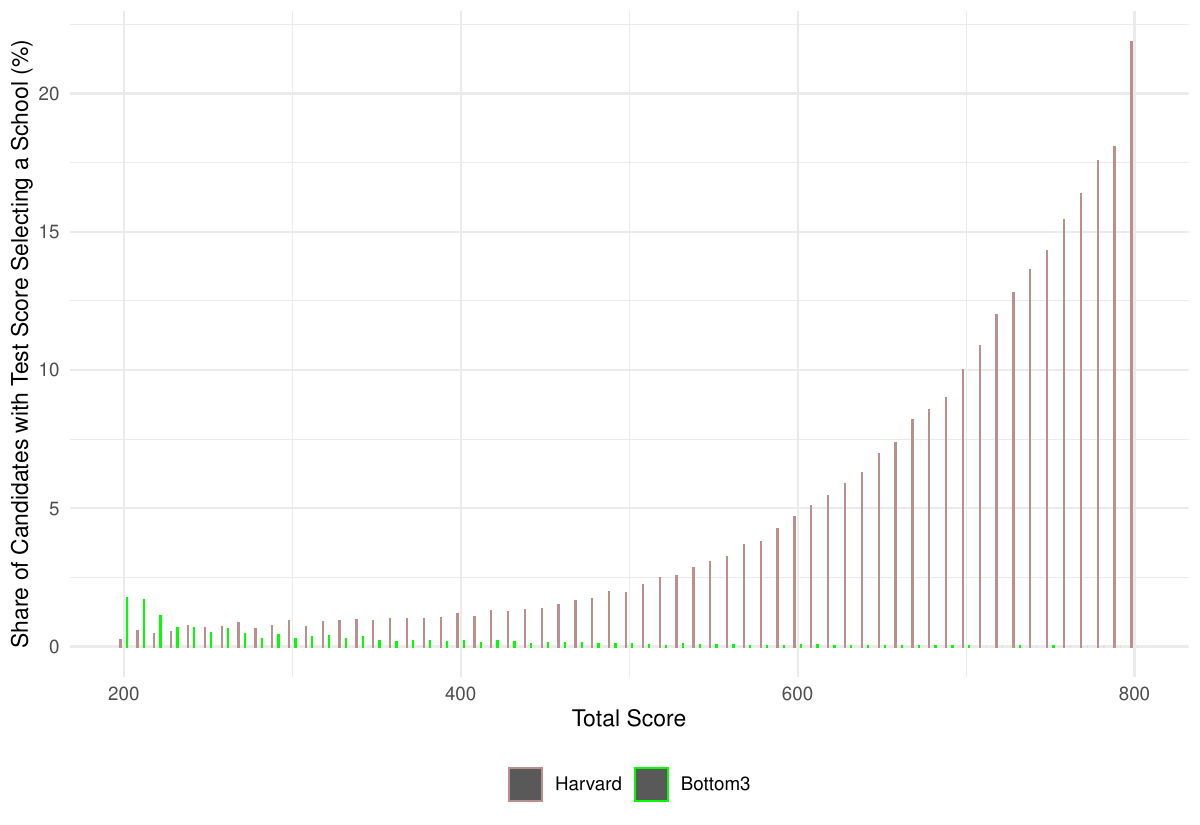}
\end{center}
\begin{minipage}[h]{.97\linewidth}\footnotesize{}{Note: The x-axis shows the range of GMAT Total Scores. The y-axis in the first panel shows the percentage of candidates in each group that received the test score on the x-axis. The first group represents all FTMBA program candidates, the second group is the subset of candidates that selected one of three low-ranked programs, the last group is the subset of candidates that selected Harvard, as one example of a top school. The second panel approximates a conditional probability distribution, showing the share of candidates with a given test score, that selected into one of the two sub-groups in panel 1.}\\
\footnotesize{}{Data Source: GMAC}
\end{minipage}
\end{figure}
\fi

The second panel of Figure \ref{fig:scoredist2} provides some evidence for the claim of monotonicity behind the measure we proposed in Section~\ref{sec:i1}. It also helps to illustrate how test scores can be useful for inferring the generic value of schools. While the first panel shows $P(S=s | c \in C)$, the second panel shows $P(C=c | s \in S)$, the fraction of students with score $s$ that select college $c$ (i.e., the empirical approximation to $g_s(c)$ in Section \ref{sec:i1}). It shows how, for a leading school, the fraction of candidates with the highest score is higher, and how the fraction of candidates decreases consistently as test scores fall, going from right to left. The set of low-reputation schools that we use as a contrast is not as monotonous. The main benefit of this plot of $P(C=c | S=s)$ is highlighting how, if we stacked the distributions for all schools, we eventually reach a uniform line of 100\% for each test score. The implication here is that if we started with a concave distribution for the top (or the group of the top) schools, the group of schools with lower reputations must, on average, have distributions that are less monotone.

The assortative matching described in Section \ref{sec:model}, and implied in related works \citep[e.g.,][]{macleod2017big} is observed to loosely hold in the data, illustrated in Figure \ref{fig:heatmap}. The heatmap in the figure represents the average scores received by schools on the vertical axis, and the average scores of candidates on the horizontal axis. The tiles are colored to show the share of candidates with each score, selecting a program in each bucket defined by a bin on the vertical axis. There are fewer candidates at the extremes of the horizontal axis (as shown in Figure \ref{fig:scoredist2}), so it is unsurprising that sorting is more obvious at the top and bottom of the score distribution. The rest of the score distribution shows the general pattern of positive assortative matching, despite the wide dispersal of candidates' selections in the middle of the score distribution.

\ifnum\figmode=0
\begin{landscape}
\begin{figure}[htbp]\caption{Heatmap of Candidate Total Scores and Program Average Scores}
\label{fig:heatmap}
\begin{center}
\includegraphics[width=0.9\columnwidth]{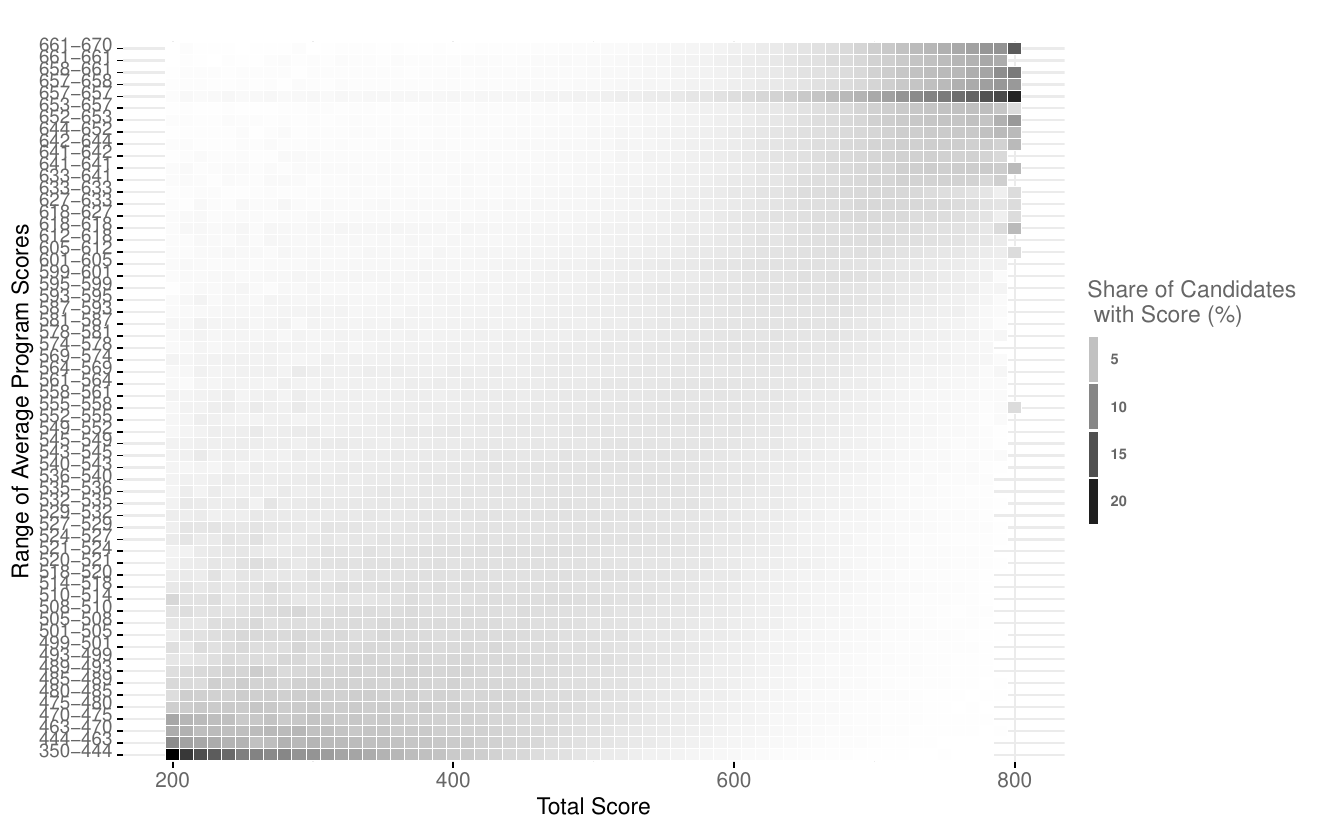}
\end{center}
\begin{minipage}[h]{.97\linewidth}
\footnotesize{}{Note: The graph shows the fraction of candidates with each test score that selected school program(s) in a score category. The score categories are defined by the average of GMAT Scores sent to the school programs. We set the number of categories to match the number of possible GMAT scores. Programs with fewer than 900 score reports are excluded.}\\
\footnotesize{}{Data Source: GMAC}
\end{minipage}
\end{figure}
\end{landscape}
\fi

\section{Results}\label{sec:res}
\subsection{The $m$ measure}
Table \ref{tab:t2} presents the rank ordering of full-time MBA Programs using the $m$ measure (based on the monotonicity approach described in Section \ref{sec:i1}). Given space constraints, only the top 35 schools are shown. The first column shows the $m$-values aggregated across all possible test scores, while the second column shows the same for the $m+$-values, which should be more robust, given that its construction is based on a logical test of whether the distribution of test scores for each college shows monotonicity, rather than the actual measure which may be sensitive to distortions of the distribution tied to features of candidates or other factors.

\ifnum\figmode=0

\begin{table}[!htbp] 
\begin{center}
  \caption{$M$ and $M+$ Measures} 
  \label{tab:t2} 
\resizebox{0.99\textwidth}{!}{
{\footnotesize \begin{tabular}{@{\extracolsep{5pt}} clrr} 
\\[-1.8ex]\hline 
\hline \\[-1.8ex] 
 & University &  $m$ value & $m+$ value \\ 
\hline \\[-1.8ex] 
1 & Harvard University & $60,840.42$ & $3,612$ \\
2 & Stanford University & $51,376.91$ & $3,588$ \\
3 & University of Pennsylvania & $41,251.79$ & $3,560$ \\
4 & Columbia University & $31,595.05$ & $3,536$ \\
5 & Massachusetts Institute of Technology (MIT) & $29,497.54$ & $3,498$ \\
6 & Northwestern University & $27,297.17$ & $3,552$ \\
7 & University of Chicago & $23,775.61$ & $3,502$ \\
8 & University of California - Berkeley & $19,813.34$ & $3,492$ \\
9 & New York University & $17,603.50$ & $3,340$ \\
10 & Duke University & $11,388.95$ & $3,352$ \\
11 & Dartmouth College & $10,736.16$ & $3,518$ \\
12 & University of California - Los Angeles & $9,972.73$ & $3,122$ \\
13 & Yale University & $9,498.43$ & $3,258$ \\
14 & University of Virginia & $8,207.52$ & $3,190$ \\
15 & University of Michigan - Ann Arbor & $8,098.09$ & $3,332$ \\
16 & University of Texas - Austin & $5,615.18$ & $2,876$ \\
17 & Cornell University & $4,467.38$ & $3,114$ \\
18 & Georgetown University & $4,204.24$ & $2,948$ \\
19 & University of Southern California & $3,638.97$ & $2,708$ \\
20 & University of North Carolina - Chapel Hill & $3,523.65$ & $2,802$ \\
21 & Carnegie Mellon University & $2,870.80$ & $3,108$ \\
22 & Northwestern University$^*$ & $2,213.34$ & $3,020$ \\
23 & University of Notre Dame & $1,729.40$ & $2,702$ \\
24 & University of Washington - Seattle & $1,555.86$ & $2,614$ \\
25 & Vanderbilt University & $1,145.81$ & $2,432$ \\
26 & Boston College & $1,077.92$ & $2,468$ \\
27 & University of Minnesota - Duluth & $1,069.98$ & $1,406$ \\
28 & Indiana University - Bloomington & $992.51$ & $2,452$ \\
29 & Boston University & $984.37$ & $2,420$ \\
30 & Emory University & $855.64$ & $2,426$ \\
31 & Brigham Young University & $852.71$ & $2,334$ \\
32 & Washington University in St. Louis & $723.16$ & $2,676$ \\
33 & University of Minnesota - Twin Cities & $657.06$ & $2,392$ \\
34 & University of Wisconsin - Madison & $584.12$ & $2,366$ \\
35 & Babson College & $569.23$ & $2,348$ \\
\hline
\end{tabular} }
}   
\end{center} 
\raggedright{\footnotesize{}{The table shows our $m$-measure, based on the fraction of students with each test score that selected a school as part of the GMAT exam, $m(c) = \sum_{s,s'\in S}(g_{s'}(c) - g_s(c))(s'-s)$. We show the composite values of $m$ as well as the $mp(c) = \sum_{s,s'\in S}\left[(g_{s'}(c) - g_s(c))(s'-s)\right]>0$ for the top 35 programs. The table is ordered by the $m$ measure. $^*$ Programs within each school are ranked separately, this line shows the one-year accelerated FTMBA program for Northwestern University, not the regular 2-year FTMBA program on line 6.}}
\end{table}
\fi

A remarkable feature of the rankings in Table~\ref{tab:t2} is how closely they resemble the existing expert rankings. The top 7 programs for column 1 in the table are consistently also the top 7 for the annual  USNWR rankings, if one aggregates the latter over nine years of data. The correlation between the top 50 and the corresponding 2018 USNWR rankings is, for the $m$ and $m+$ measures, 0.92 and 0.94 respectively.

The two measures that we report in Table~\ref{tab:t2} are remarkably consistent, suggesting that our method is robust. The rank correlation for the $m$-measure and the $m+$-measure is 0.68 (which increases to more than 0.9 when limited to the top 50). The correlation indicates that the rankings are not unduly sensitive to the selected $m$-measure. The robustness of the results comes in part from using the sum of the measure across all $s$ reference scores. In theory, the reference score used should not matter if all college conditional score distributions were strictly monotone, but in practice, summing across reference scores reduces noise as the number of reports gets smaller and the distributions less strictly monotone. 

Nevertheless, we attempt another approach to estimating the underlying ranking that is more computationally intensive, but is potentially a robustness check on the results in Table \ref{tab:t2}.

\subsection{Tournament Ranking}

Table \ref{tab:t5} presents the rank ordering of full-time MBA Programs, using the tournament approach given in Section \ref{sec:i2}. Specifically, we counted the number of times, in a hypothetical `matchup' between any two schools, candidates with the higher test score in a given year selected one school, but not the other one. In principle, given the full set of school choice portfolios by candidates, a vector of length $\sum_{i \in I} k_i$ (because each candidate $i$ selects a portfolio of length $k_i$), we construct the set of pairings of length $(\sum_{i \in I} k_i)^2$. Convolving all portfolio choices, limited to choices within the same GMAT test year for tractability, effectively lists every opportunity students had to choose either the first or second school in a pairing. We can then identify the `winner' of each pairing as the one with a greater number of instances where candidates with a higher score selected that school but not the other one. The last column of the table (labeled wins) shows the number of `match-ups' out of the possible 687 that each school won. For example, for Stanford's FTMBA program, out of its 687 match-ups in the tournament, 684 times more high-scoring portfolios in the same (July-June) year included Stanford without selecting the other school than the converse. Candidates that selected both schools in a match-up were dropped from the count.

\ifnum\figmode=0
\begin{table}[!htbp] 
\begin{center}
      \caption{School Programs Tournament Rank}
  \label{tab:t5}
\resizebox{0.85\textwidth}{!}{
{\footnotesize \begin{tabular}{@{\extracolsep{5pt}} clr}
\\[-1.8ex]\hline
\hline \\[-1.8ex]
Rank & University & wins \\
\hline \\[-1.8ex]
1 & Harvard University & $687$ \\
2 & Massachusetts Institute of Technology (MIT) & $686$ \\
3 & University of Pennsylvania & $685$ \\
4 & Stanford University & $684$ \\
5 & University of Chicago & $683$ \\
6 & Northwestern University & $682$ \\
7 & Dartmouth College & $681$ \\
8 & Columbia University & $680$ \\
9 & University of California - Berkeley & $679$ \\
10 & Yale University & $678$ \\
11 & Duke University & $677$ \\
12 & University of Virginia & $676$ \\
13 & University of Michigan - Ann Arbor & $675$ \\
14 & Cornell University & $674$ \\
15 & Carnegie Mellon University & $673$ \\
16 & Cornell University$*$ & $672$ \\
17 & New York University & $671$ \\
18 & Northwestern University$*$ & $670$ \\
19 & University of California - Los Angeles & $669$ \\
20 & Georgetown University & $668$ \\
21 & University of Texas - Austin & $667$ \\
22 & University of North Carolina - Chapel Hill & $666$ \\
23 & University of Southern California & $664$ \\
24 & Washington University in St. Louis & $663$ \\
25 & University of Notre Dame & $663$ \\
26 & Brandeis University & $662$ \\
27 & Emory University & $660$ \\
28 & University of Washington - Seattle & $659$ \\
29 & University of Michigan - Ann Arbor$**$ & $659$ \\
30 & Boston College & $657$ \\
31 & Indiana University - Bloomington & $656$ \\
32 & Boston University & $655$ \\
33 & Duke University$**$ & $655$ \\
34 & Vanderbilt University & $654$ \\
35 & University of Wisconsin - Madison & $653$ \\
\hline \\[-1.8ex]
\end{tabular} }
}
\end{center}
\raggedright{\footnotesize{}{The table shows the top 35 FTMBA programs from a tournament ranking approach. All 688 MBA programs are matched against one another, with programs notching a point for each candidate that selected the program but did not select the other one, when at least one person with a higher score selected the other program. Having more points in each of the 687 match-ups counts as a win for a school program. Section \ref{sec:i2} explains the details of the match-ups.
$^*$ shows one-year accelerated FTMBA programs, (not the regular 2-year FTMBA), while $^{**}$ shows the Cross-Continent/Global MBA program for the University.}}
\end{table}
\fi

Table \ref{tab:t5} offers an improvement over a ranking based simply on how many times candidates selected a school, because the simple count does not capture how candidates weigh the selectivity of schools. The approach also represents an improvement over taking school average test scores of accepted candidates, as schools can bias rankings in their favor by selecting only candidates with high test scores. Even if one used the average tests of all candidates that selected a school program, there is the risk for programs receiving a low number of test scores, that the occasional idiosyncratic selection of the program by a high-scoring candidate distorts or adds noise to the measure. 

The 
approach in Table \ref{tab:t5} builds on the model in Section \ref{sec:model}, 
and reflects our theoretical connection between
high-scoring candidates expecting better chances of admission and therefore choosing to apply to more selective schools. 
We can learn how candidates value schools from comparing the schools candidates selected, against the ones that candidates did not select, while filtering out score reports that, on average, do not have the same level of confidence in their options for admission.  

The top school programs in Table \ref{tab:t5} also closely resemble the rankings from the expert opinion sources, as well as the rankings from the $m$ measures (Section~\ref{sec:i1}). The Spearman rank correlation coefficient between the ranking in Table~\ref{tab:t5} and the 2018 USNWR ranking for the top 50 schools is 0.88, and the corresponding coefficients for the two $m$-measures in Table \ref{tab:t2} are 0.92 and 0.94 respectively. Outside the top 50, the correlation between the different ranking systems drops off.\footnote{Note that because the USNWR reports rankings for universities, not specific programs, we keep only the standard 2-year full-time MBA, one-per-university, for the estimates of correlation with the USNWR rankings.}

A major advantage for revealed preference rankings based on candidate selections is that we can rank all programs with little effort, not just the top 50 or 100, as is usual for the US News Report and other sources of expert opinion. Another advantage, as shown in previous work \citep[c.f.,][]{averyetalRPranking}, is that we can go beyond estimates of the value of a school to a generic student, by tailoring the approaches to subsets of student candidates that may have reasons to value schools differently.

\subsection{Beyond Generic Rankings}
Table \ref{tab:maj} shows separate $m$ measure-rankings for two sub-groups of candidates: business/economics majors and other majors. We can rationalize this categorization. While MBA enrollment is almost uniformly motivated by the prospects of better careers and incomes, it is a different kind of career transition for business and economics majors, who usually come from jobs in accounting or consulting roles, compared with engineering majors or others for whom the MBA is a transition from operational roles to management or to consulting roles. There is no reason \emph{ab initio} to expect different rankings for these two groups, given how all well-known schools offer general business administration as the most common track for admitted students. Nevertheless, where schools offer specializations in marketing, accounting, finance or other functions that are easier for economics and business majors to take, candidates from those majors may reveal a preference for schools with those specialization tracks.\footnote{The 42 majors that GMAC allows candidates to choose, in describing their undergraduate degrees include: Accounting, Finance, Management, Economics and six other related majors that we code as business/economics majors. The remaining, with Engineering as the largest sub-group, but including blanks, and selections of ``My major or field of study is not shown here" are coded as Other.}


\ifnum\figmode=0
\begin{table}[!htbp]
    \begin{center}
    \caption{Rankings by Undergrad Major Categories}
\resizebox{0.99\textwidth}{!}{
\begin{tabular}{@{\extracolsep{5pt}} clrrc}
\\[-1.8ex]\hline
\hline \\[-1.8ex]
 & University & Business-Major & Other & Rank 2 \\
\hline \\[-1.8ex]
1 & Harvard University & $54,895.85$ & $60,355.61$ & $1$ \\
2 & Stanford University & $51,829.62$ & $50,819.01$ & $2$ \\
3 & University of Pennsylvania & $41,695.77$ & $39,719.20$ & $3$ \\
4 & Columbia University & $30,825.70$ & $30,616.29$ & $5$ \\
5 & Massachusetts Institute of Technology (MIT) & $27,520.64$ & $32,092.05$ & $4$ \\
6 & Northwestern University & $26,933.38$ & $26,251.89$ & $6$ \\
7 & University of Chicago & $26,280.87$ & $21,219.54$ & $7$ \\
8 & New York University & $23,097.43$ & $16,029.62$ & $9$ \\
9 & University of California - Berkeley & $22,923.71$ & $20,524.71$ & $8$ \\
10 & Duke University & $10,441.33$ & $11,335.77$ & $10$ \\
11 & Dartmouth College & $10,259.33$ & $10,478.46$ & $11$ \\
12 & University of California - Los Angeles & $10,029.10$ & $9,312.07$ & $13$ \\
13 & University of Virginia & $9,498.72$ & $7,232.24$ & $15$ \\
14 & Northwestern University$^*$ & $8,789.46$ & $781.92$ & $28$ \\
15 & Yale University & $8,640.99$ & $9,791.30$ & $12$ \\
16 & University of Michigan - Ann Arbor & $6,556.73$ & $8,377.74$ & $14$ \\
17 & University of Texas - Austin & $5,630.72$ & $5,543.51$ & $16$ \\
18 & Georgetown University & $4,714.83$ & $3,394.81$ & $18$ \\
19 & University of North Carolina - Chapel Hill & $4,153.64$ & $2,725.33$ & $21$ \\
20 & Cornell University & $4,144.41$ & $4,675.49$ & $17$ \\
21 & University of Southern California & $3,974.11$ & $2,975.03$ & $20$ \\
22 & University of Notre Dame & $2,189.31$ & $1,310.22$ & $24$ \\
23 & Carnegie Mellon University & $2,018.70$ & $3,387.08$ & $19$ \\
24 & Vanderbilt University & $1,609.70$ & $638.18$ & $30$ \\
25 & University of Washington - Seattle & $1,382.98$ & $1,455.37$ & $22$ \\
26 & Indiana University - Bloomington & $1,238.08$ & $605.37$ & $31$ \\
27 & Brigham Young University & $1,236.05$ & $323.79$ & $36$ \\
28 & Boston College & $1,079.23$ & $953.53$ & $26$ \\
29 & Emory University & $1,076.24$ & $546.23$ & $33$ \\
30 & University of Wisconsin - Madison & $806.91$ & $328.44$ & $35$ \\
31 & University of Minnesota - Twin Cities & $741.67$ & $518.69$ & $34$ \\
32 & Washington University in St. Louis & $726.56$ & $729.35$ & $29$ \\
33 & Boston University & $553.75$ & $1,266.39$ & $25$ \\
34 & University of Rochester & $506.49$ & $230.83$ & $37$ \\
35 & Southern Methodist University & $381.45$ & $-229.04$ & $275$ \\
\hline \\[-1.8ex]
\end{tabular}
}
    \label{tab:maj}
    \end{center}
\raggedright{\footnotesize{}{
The table shows our $m$ measure, based on the schools students selected, calculated separately for broad groups of undergraduate majors. Business majors, the first column,  includes economics, and all business undergraduate degrees, while the `Other` column holds all other undergraduate majors. We show the composite values for the programs ranked as top 35. The table is ordered by the values of the Business majors group. The last column shows the rank for the category holding other undergrad majors. $^*$ This line refers to the one-year accelerated FTMBA program for the school. See Table \ref{tab:t2} notes.}}
\end{table}
\fi

The Spearman's rank correlation coefficient for the different rank orderings in Table \ref{tab:maj} is 0.93, implying that the rankings for business and econ majors broadly resemble those obtained for candidates with other undergraduate academic backgrounds (using the approach in Section \ref{sec:i1}). 
A more detailed grouping of majors is used to prepare the rankings in Table \ref{tab:tall} in the Appendix. Using ten groups instead of 2 shows the broad pattern of similarities in how different groups of students value schools, as well as the notable differences -- between economics and engineering, for example.

Table \ref{tab:tctzn} shows differences in the school rankings for international students compared to US citizens. To be clear, the data we use show only tests taken by US-based students, but of these students, about 18.4\% are non-citizens. Just as for the table showing separate rankings for econ/business majors vs. others, this table shows how international students and citizens see differences in the reputations of schools. The last column shows the ranking generated for non-citizens. 

\ifnum\figmode=0
\begin{table}[!htbp]
    \begin{center}
    \caption{Rankings: Citizens and Internationals}
\resizebox{0.99\textwidth}{!}{
\begin{tabular}{@{\extracolsep{5pt}} clrrc}
\\[-1.8ex]\hline
\hline \\[-1.8ex]
 & University & US Citizens & Noncitizens & Rank 2 \\
\hline \\[-1.8ex]
1 & Harvard University & $60,580.74$ & $54,401.89$ & $1$ \\
2 & Stanford University & $51,069.98$ & $46,000.91$ & $3$ \\
3 & University of Pennsylvania & $39,431.28$ & $46,980.27$ & $2$ \\
4 & Columbia University & $29,992.23$ & $35,849.85$ & $4$ \\
5 & Massachusetts Institute of Technology (MIT) & $27,995.47$ & $32,292.72$ & $5$ \\
6 & Northwestern University & $27,360.08$ & $24,555.17$ & $7$ \\
7 & University of Chicago & $22,620.66$ & $27,787.66$ & $6$ \\
8 & University of California - Berkeley & $19,559.78$ & $19,939.65$ & $8$ \\
9 & New York University & $16,586.46$ & $18,910.66$ & $9$ \\
10 & Duke University & $11,409.93$ & $9,806.70$ & $11$ \\
11 & Dartmouth College & $11,084.70$ & $8,125.19$ & $13$ \\
12 & University of California - Los Angeles & $10,326.46$ & $8,480.00$ & $12$ \\
13 & Yale University & $9,323.67$ & $10,488.48$ & $10$ \\
14 & University of Virginia & $9,036.48$ & $5,036.51$ & $17$ \\
15 & University of Michigan - Ann Arbor & $7,900.16$ & $7,396.70$ & $14$ \\
16 & University of Texas - Austin & $5,992.27$ & $3,781.67$ & $18$ \\
17 & Georgetown University & $4,866.80$ & $1,687.01$ & $23$ \\
18 & Cornell University & $4,046.25$ & $6,530.12$ & $15$ \\
19 & University of Southern California & $3,823.35$ & $2,603.73$ & $20$ \\
20 & University of North Carolina - Chapel Hill & $3,428.01$ & $2,787.93$ & $19$ \\
21 & Carnegie Mellon University & $2,286.97$ & $5,670.38$ & $16$ \\
22 & University of Notre Dame & $1,980.46$ & $679.57$ & $30$ \\
23 & Northwestern University$^*$ & $1,932.38$ & $2,191.48$ & $21$ \\
24 & University of Washington - Seattle & $1,859.85$ & $686.09$ & $29$ \\
25 & Vanderbilt University & $1,366.04$ & $42.27$ & $51$ \\
26 & Boston University & $1,265.37$ & $477.63$ & $32$ \\
27 & Brigham Young University & $1,260.17$ & $-916.38$ & $555$ \\
28 & Boston College & $1,137.64$ & $862.89$ & $27$ \\
29 & University of Minnesota - Duluth & $1,113.18$ & $-136.19$ & $252$ \\
30 & Indiana University - Bloomington & $973.45$ & $1,300.68$ & $25$ \\
31 & Washington University in St. Louis & $877.13$ & $384.08$ & $35$ \\
32 & University of Minnesota - Twin Cities & $851.28$ & $164.90$ & $42$ \\
33 & University of Wisconsin - Madison & $700.24$ & $105.37$ & $46$ \\
34 & University of Colorado Boulder & $664.94$ & $-727.57$ & $526$ \\
35 & Babson College & $651.89$ & $330.02$ & $38$ \\

\hline \\[-1.8ex]
\end{tabular}
}
    \label{tab:tctzn}
    \end{center}
\raggedright{\footnotesize{}{
The table shows our $m$ measure calculated separately for US citizens and international students. The table is ordered by the values for citizens. The last column shows the rank for international students. $^*$ This line refers to the one-year accelerated FTMBA program for the school. See Table \ref{tab:t2} notes.}}
\end{table}
\fi

There are good reasons to expect international students to evaluate the reputations of business schools differently. First, as non-citizens, many who want to work in the US after their studies are seeking locations with clusters of employers that hire. Larger urban centers with an international reputation may also be favored by non-citizen students, who may want degrees from institutions that are known in their home countries. 

The results in Table \ref{tab:tctzn} reflect the suggested differences in interest. The Spearman's rank correlation between the ordering of schools by citizens and non-citizens is 0.71. Much of the differences in the rankings occur outside of the top 15, arguably because of the strong overlap in shared interests across the two groups for the most highly ranked schools. 

The instances with large differences in rankings are instructive, in suggesting the kinds of schools that appeal or do not appeal to foreign students. The three schools ranked outside of the top 100 by non-citizens, but in the top 35 by citizens, are all located in small cities away from major metropolises. These differences in rankings between citizens and non-citizens are another advantage of using approaches like in this paper, that allow for deviations from a generic ranking on schools, while revealing the preferences of past candidates from their choices and scores.

To address the idea that school reputations change over time (which arguably explains why expert-opinion school rankings are updated annually), Table \ref{tab:prd} shows rankings for tests taken before July 1, 2011 (Early period) compared with tests taken after (Later period). The simplified grouping into two periods can be generalized into narrower time-splits, but for this paper, two periods is enough to illustrate the principle. As in previous Tables, we use the approach in Table \ref{tab:t2} to check for the monotonicity of selections by score for subsets of the data.

\ifnum\figmode=0
\begin{table}[!htbp]
    \begin{center}
    \caption{Rankings: Early and Later Periods}
\resizebox{0.99\textwidth}{!}{
\begin{tabular}{@{\extracolsep{5pt}} clrrc}
\\[-1.8ex]\hline
\hline \\[-1.8ex]
 & University & Later & Early & Rank 2 \\
\hline \\[-1.8ex]
1 & Harvard University & $61,625.55$ & $62,316.47$ & $1$ \\
2 & Stanford University & $53,594.05$ & $50,939.35$ & $2$ \\
3 & University of Pennsylvania & $41,243.99$ & $41,514.80$ & $3$ \\
4 & Massachusetts Institute of Technology (MIT) & $33,887.49$ & $26,017.71$ & $5$ \\
5 & University of Chicago & $28,451.34$ & $20,297.96$ & $7$ \\
6 & Columbia University & $28,422.28$ & $35,225.33$ & $4$ \\
7 & Northwestern University & $27,817.83$ & $25,461.07$ & $6$ \\
8 & University of California - Berkeley & $20,807.98$ & $18,655.18$ & $9$ \\
9 & New York University & $15,737.72$ & $18,846.39$ & $8$ \\
10 & Duke University & $13,324.78$ & $9,575.18$ & $11$ \\
11 & Dartmouth College & $12,352.50$ & $9,118.52$ & $13$ \\
12 & Yale University & $9,570.94$ & $9,423.80$ & $12$ \\
13 & University of Michigan - Ann Arbor & $8,907.03$ & $6,913.51$ & $15$ \\
14 & University of California - Los Angeles & $8,547.90$ & $10,689.26$ & $10$ \\
15 & University of Virginia & $8,476.41$ & $7,938.16$ & $14$ \\
16 & University of Texas - Austin & $5,856.30$ & $5,402.97$ & $16$ \\
17 & Cornell University & $3,971.39$ & $4,795.05$ & $17$ \\
18 & Georgetown University & $3,581.05$ & $4,524.20$ & $18$ \\
19 & Northwestern University$^*$ & $3,322.52$ & $904.00$ & $29$ \\
20 & University of North Carolina - Chapel Hill & $3,222.41$ & $3,720.43$ & $20$ \\
21 & University of Southern California & $3,015.49$ & $3,984.26$ & $19$ \\
22 & Carnegie Mellon University & $2,738.56$ & $2,925.25$ & $21$ \\
23 & Vanderbilt University & $1,720.93$ & $842.10$ & $30$ \\
24 & University of Washington - Seattle & $1,632.21$ & $1,496.33$ & $24$ \\
25 & University of Notre Dame & $1,538.99$ & $1,820.17$ & $23$ \\
26 & Emory University & $1,280.38$ & $637.33$ & $34$ \\
27 & Indiana University - Bloomington & $1,032.00$ & $967.72$ & $28$ \\
28 & Boston College & $904.16$ & $1,219.73$ & $25$ \\
29 & University of Minnesota - Twin Cities & $775.52$ & $588.47$ & $35$ \\
30 & Boston University & $752.64$ & $1,177.05$ & $26$ \\
31 & Washington University in St. Louis & $737.05$ & $740.48$ & $32$ \\
32 & Brigham Young University & $645.98$ & $995.83$ & $27$ \\
33 & Cornell University$^**$ & $492.16$ & $567.24$ & $36$ \\
34 & University of Wisconsin - Madison & $380.31$ & $761.35$ & $31$ \\
35 & University of North Carolina - Chapel Hill & $363.10$ & $-34.46$ & $94$ \\
\hline \\[-1.8ex]
\end{tabular}
}
    \label{tab:prd}
    \end{center}
\raggedright{\footnotesize{}{
The table shows our $m$ measure  calculated separately for tests taken before July 1, 2011 (Early) and after (Later). The table is ordered by the values for more recent scores. The last column shows the rank for the early period. $^*$ This line refers to the one-year accelerated FTMBA program for the school. See Table \ref{tab:t2} notes.}}
\end{table}
\fi

The results in Table \ref{tab:prd} imply, as seen for rankings like the USNWR, that the reputations of schools are relatively stable over time. The Spearman's rank correlation coefficient between the two periods is 0.86, which suggests some differences over time, but no major realignment of the order of school reputations. The most notable move is the climb of UNC Chapel Hill from 94 in the early period to 35 in the later period. The few cases getting worse rankings in the later period are slight losses of five places (by Brigham Young), four places (from 26 to 30 by Boston University, 10 to 14 by UCLA), three places (Boston College from 25 to 28, U Wisconsin from 31 to 34) as well as other drops by two places or one in the rankings.

\subsection{Robustness Checks}
Given that about 22\% of candidates in our data took the test more than once, Table \ref{tab:t4} repeats Table \ref{tab:t2}, using the selections made during the attempt with the best score for candidates with more than one test attempt. The aggregate ranking is based on the method in Table \ref{tab:t2}, summing the $m$-measures created for each school's program across all the possible scores (limited to the score reports for each candidates' identified best attempt). Once again, the ranking resembles the outputs in Table \ref{tab:t2} (with a correlation coefficient of 0.98) as well as commonly-used expert rankings. 

The orderings do not match exactly across the two tables. The differences imply two things: First, that as students learn about their desirability to schools through test-taking, the portfolio of schools that they select changes, as shown in related papers \citep[e.g.,][]{alishorrer24,goodman2016learning}. Other admission tests have a similar pattern of candidates retaking tests \citep{goodman2020take}. Second, the updating of school portfolios by candidates as they take the GMAT exam appears to capture both candidates expectations of the selectiveness of schools, based on test scores, as well the candidates' own idiosyncratic preferences, which makes the case for rankings tailored to subgroups of candidates.

\section{Discussion}\label{sec:discussion}

A student faces a portfolio choice problem of choosing $k$ schools to select at the application stage. We show that the optimal portfolio choice with a budget of $k$ applications has a recursive structure and may be represented as
\[V(k) = \max\{p(c)v(c) + (1-p(c)) V(k-1):c\in C \}.
\] Here $V(k)$ and $V(k-1)$ are the values (or indirect utilities) of choosing an optimal portfolio with, respectively, $k$ and $k-1$ schools. The recursive nature of the problem means that a student can choose an optimal $k$th school that is at least as good as any of the schools in the optimal $k-1$ portfolio, and then use the solution to the $k-1$-portfolio problem to find the remaining schools to apply to.

Figure~\ref{fig:theory} illustrates the problem of choosing the $k$th optimal school, given that $V=V(k-1)$. The diagram has the probability of admissions of each school $c$, $p(c)$, on the horizontal axis, and the utility from the school, $v(c)$, on the vertical axis. The figure assumes (but our theory does not) that the set of feasible schools is represented on a smooth concave curve. The indifference curves of a utility function $(p,v)\mapsto pv+(1-p)V$ are represented in violet. Indeed, we only draw a single indifference curve: the one that is tangent to the blue curve of feasible probability -- utility points.

\ifnum\figmode=0
\begin{figure}[h]
\begin{tikzpicture}
    \begin{axis}[
        axis lines=middle, 
        xmin=-0.05, xmax=3.5, 
        ymin=-0.05, ymax=3.5,  
        samples=100,
        domain=0:3,
        xlabel={\footnotesize $p(c)$},xlabel style={below},
        ylabel={\footnotesize $v(c)$},ylabel style={left},
        clip=false,
        xtick=\empty, ytick=\empty  
    ]
        \addplot[blue, thick, domain=0:sqrt(5)] {max(0, 2 - 0.4*x^2)}; 
        
        \pgfmathsetmacro\xone{0.6}
        \pgfmathsetmacro\xtwo{1.6}
        \pgfmathsetmacro\fone{2 - 0.4*\xone^2}
        \pgfmathsetmacro\ftwo{2 - 0.4*\xtwo^2}
        \pgfmathsetmacro\slopeone{-0.8*\xone}
        \pgfmathsetmacro\slopetwo{-0.8*\xtwo}
        
        \addplot[red, dashed, domain=0:1.4] {max(0, \fone + \slopeone*(x-\xone) + 0.2*(x-\xone)^2)} node[above right] {\footnotesize $p(c)v(c) + (1-p(c)) V'$};
        \addplot[violet, dashed, domain=1.0:2.4] {max(0, \ftwo + \slopetwo*(x-\xtwo) + 0.2*(x-\xtwo)^2)} node[above right] {\footnotesize $p(c)v(c) + (1-p(c)) V$};
        
        \addplot[black, only marks, mark=*] coordinates {(\xone,\fone) (\xtwo,\ftwo)};
    \end{axis}
\end{tikzpicture}
\caption{Choosing the $k$th school.}\label{fig:theory}
\end{figure}
\fi

Key to our results is a comparative statics property. If the value of choosing the ``remaining'' schools increases from $V$ to $V'>V$, then the value of the chosen school for the $k$th position increases while the probability of admission decreases. The diagram in Figure~\ref{fig:theory} illustrates our comparative statics argument as a shift from the red to the violet indifference curve: An increase in $V'$ means that the indifference curves in $(p,v)$ space become steeper.  The new tangency point then occurs for a lower value of $p$ and a larger value of $v$.

Why is the comparative statics property important? It translates into a comparative statics property of the value chosen for the school when a student increases their test score. A student with a larger test score will recursively have larger utilities under their optimal portfolios: the values of $V(1), V(2),\ldots, V(k)$. These will translate into a first-order stochastically dominating change in the utilities $v$ from each school. Because such increases in utilities must be associated with lower values of $p$, they must translate into a first-order stochastically dominating shift in selectivity of the school in the optimal portfolio. Finally, the comparative statics says that students with better test scores will, \textit{ceteris paribus}, apply more aggressively --- to more selective schools --- than will  students with lower test scores. This comparative statics is at the heart of our Theorem~\ref{thm:main} and the methods in our paper.

At this point, with the previous argument as background, it is fruitful to discuss the assumption of duplicate colleges. In a model with finitely many schools, the diagram in Figure~\ref{fig:theory} would feature a punctuated version of the blue line. Instead of a continuous curve, we would see finitely many pairs $(p,v)$ arranged in a decreasing sequence.\footnote{Of course one could have $(p,v)$ and $(p',v')$ with $p<p'$ and $v'<v$ but such dominated schools would not be chosen by the particular student in question, and can safely be ignored. Such schools might not be dominated for a student with a different utility function, but this consideration is not relevant for the current discussion which is centered on the choices of a particular individual student.} If these finitely many points were sparse enough, and the increase in the slope of the indifference curve were small enough, then it is possible that the new optimal point is at the same college as before. If that were to happen, the student would choose to apply to the same college twice --- which would not make sense unless there are duplicate colleges. Hence our assumption of college duplicates. At the same time, the figure and our argument also illustrate that such a situation is most likely empirically irrelevant. 

Finally, we discuss some of the strengths of our model. We make no assumptions on $\la$, which translates into completely flexible unobserved heterogeneity on the part of students. There is presumably a common value component to students' preferences, but this is exclusively captured by the order $\leq$ on colleges. So $c<c'$ if it is harder for any student to get into $c'$ than into $c$. The achievement in Theorem~\ref{thm:main} is that we are able to identify (or partially identify) this order purely from the application data by score.

The model makes strong assumptions on the heterogeneity of admission thresholds and likelihoods. We assume a common distribution $F$ across colleges.

\section{Conclusion}\label{sec:conc}
School reputations matter in the grand scheme of how universities serve society at large. When candidates compete for places with the highest single-index reputation ranks, several things happen: universities face perverse incentives to dress up their rankings, so that they are seen as providing greater value to prospective students, at the same time, reputations -- fairly earned or not, lead to candidates self-sorting into schools. Others have described the outcome as the {\em `Big Sort'}, with high-scoring students choosing higher-ranked universities, which leads to higher-paying jobs and wage growth in the long-run for the selected few. This stratification in part reflects reputational concerns (not just a preference for good peers and positive learning externalities). There is a real risk that the perverse incentives faced by schools, and the real power that the experts creating the rankings have to select the measures that go into the single-index rank, clouds the information acquisition process for candidates considering schools.

We point to a rich information set that is relatively free of bias in estimating the perceived value of schools to prospective candidates -- the programs selected by past candidates. The revealed preference approach in this paper combines both the selections of schools by candidates, as well as the scores of the candidates. The latter reflects the subtle impact of how high-scoring students expect to have a higher probability of being accepted {\em at any school} compared with low-scoring students. Our first approach generated rankings from the monotonicity of schools' test score distributions, which we call the $m$-measure, while the second was a skewed tournament where each matchup between schools was limited to candidates with higher scores selecting the other school versus candidates with a given score selecting a school. The test scores, again, served here as a subtle indication of the expectation of being accepted at a school. 

The rankings created by the revealed preference approach largely resemble the expert rankings in the top 50, with correlations of 0.92 and 0.94 reported for the USNWR rankings, for example. For the top 100 available for USNWR, the correlation is 0.72. The rankings are relatively immune to the choices of school administrators that can bias other rankings, given how they simply represent an adjusted version of a `wisdom-of-the-crowd'. 

There is great potential in giving students the option of tailored rankings for their subgroup, as we do in Tables \ref{tab:maj} and \ref{tab:tctzn}. This kind of tailored ranking is a valuable feature of our approach, and one that is rooted solidly in the economic rationale that certain groups of potential students value some schools more. The range of options for defining subgroups to use in ranking customization is broad and can include college major, geography, and citizenship, among others. 

Furthermore, the revealed preference rankings presented here have the added advantage of showing more than just the top 100 schools. Every school that received a useful number of score reports can be ranked. The expert-opinion sources rationally limit their rankings to the top, as their ranking process requires volumes of data. The revealed preference approach uses data that test administrators already have. For the same reason, this approach also removes the need for a specialized survey, like the one used in \cite{averyetalRPranking}.

\section{Proof of Theorem~\ref{thm:main}}\label{sec:proof}

Fix a utility function $v$ and a score $s\in \S$. Define $U$ from $v$ as in the description of our model.
Let $V:[M] \to \Re$ and $A:[M]\to 2^C$ satisfy 
\begin{align*}
    V(K) & = \max\{ U(A): A\subseteq C \text{ and } \abs{A}\leq K \} \\
    & = U(A(K)).
\end{align*} 

Write $p(c) = F(s-t_j)$ for $c=c_j$.

\begin{lemma}\label{lem:VandA} 
The functions $V$ and $A$ satisfy \begin{align*}
V(1) & = \max \{ p(c) v(c) : c\in C \}   \\
V(K) & = \max \{ p(c) v(c) + (1-p(c)) V(K-1): c\in C \},   K>1\\
A(K) & = \{ c^1,\ldots, c^K \} 
\end{align*} with \[ 
v(c^1)\leq v(c^2) \ldots \leq  v(c^K)
\]
\end{lemma}
\begin{proof}
First, we may restrict attention to the subset $\C'$ of $c\in \C$ that satisfies the condition that there is no $c'\in\C$ with  $v(c)< v(c')$ and $p(c') \geq  p(c)$, because if there exists such $c'$ then $c$ is never chosen. So $\C'$ is the set of \df{undominated colleges}.


The proof is by induction.

The formula for $V(1)$ is obvious. Let $A(1) = \{c^1\}$. Suppose that $A(2) = \{c,c'\}$ and say wlog that $v(c)\leq v(c')$. This means that $V(2) = p(c') v(c') + (1-p(c')) p(c) v(c)$. Now observe that we must have $p(c) v(c)= \max\{p(\tilde c)v(\tilde c):\tilde c\in C\} = V(1)$ because if there is some $\tilde c'$ with $p(\tilde c)v(\tilde c)> p(c) v(c) $ then by substituting $\tilde c$ (or its duplicate) for $c$ in $A(1)$ would lead to a higher utility (the student would obtain more utility by planning to accept $\tilde c$ upon being rejected from $c'$). Thus $V(2) = p(c') v(c') + (1-p(c')) V(1)$ and we may as well assume that $c=c^1$ as $c$ must be a duplicate of $c^1$. Then, again making use of duplicates, we have that  $V(2) \geq  p(\tilde c) v(\tilde c) + (1-p(\tilde c)) V(1)$ for all $\tilde c\in \C$, and thus $V(2) =\max\{  p(\tilde c) v(\tilde c) + (1-p(\tilde c)) V(1) :\tilde c\in \C \}$. Since $A(2) = \{c,c'\}$ we can set $c^2=c'$, and we have $v(c^1)\leq v(c^2)$. 

Now let $K\geq 3$ and suppose that the lemma is true for $k\leq K-1$. Then it is easy to see that, as before,  
$V(K)  = \max \{ p(\tilde c) v(\tilde c) + (1-p(\tilde c) ) V(K-1): \tilde c\in C \}$. Let $c^K$ be such that 
$V(K)  = p(c^K) v(c^K) + (1-p(c^K)) V(K-1)$. We want to show that $v(c^{K-1})\leq v(c^K)$.

Observe that if $h$ is a (strictly) decreasing function then the function \[(x^1,x^2)\mapsto h(x^2) x^2 + (1-h(x^2)) x^1\] has (strictly) increasing differences. Among $c,c'\in \C'$, $v(c)<v(c')$ implies that $p(c)> p(c')$. So we may conclude that $p(v^{-1}(w))w+(1-p(v^{-1}(w))) V$ has increasing differences in $(w,V)$. Since $V(K-2) < V(K-1)$,  we conclude by monotone comparative statics (see \cite{milgrom1994monotone}) that $v(c^{K-1})\leq v(c^K)$.
\end{proof}

In light of Lemma~\ref{lem:VandA}, when  $C_s(v) = \{c^1,\ldots,c^K\}$ and $v(c^1)\leq v(c^2) \ldots \leq  v(c^K)$ we can think of a student applying to schools sequentially and then choosing where to go in the reverse order of application, in the event that they get admitted. If they would only apply to one school, it would be $c^1$. Then, if given the opportunity to add a school, they would apply to $c^2$; and so on.

\begin{lemma}\label{lem:paloapalo}
Fix a utility $v$ and write $K$ for $k^v_s$ and $K'$ for $k^v_{s'}$. 
Let $s<s'$, $C_s(v) = \{c^1,\ldots,c^K\}$, and $C_{s'}(v) = \{d^1,\ldots,d^{K'}\}$ with 
\[ 
v(c^1)\leq v(c^2) \ldots \leq  v(c^K) \text{ and } v(d^1)\leq v(d^2) \ldots \leq  v(d^{K'})
\]
and $K\leq K'$. Then $v(c^k)\leq v(d^k)$ for $1\leq k\leq K$.
\end{lemma}

\begin{proof} Fix $v$ and $k\in [K]$. Let $c^k\in C_s(v)$, using the notation from Lemma~\ref{lem:VandA}. So $c^k$ is chosen for the $k$th position by a student with utility $v$ and score $s$. From Lemma~\ref{lem:VandA} we know that $c^k$ solves the problem of maximizing the function
\[ \tilde{c}_j\mapsto 
F(s- \tilde{t}_j) v(\tilde{c}_j) + (1-F(s- \tilde{t}_j)) V = F(s- \tilde{t}_j) (v(\tilde{c}_j) - V) + V,
\] where $V= V(k-1)$.
Denote by $V'$ the value $V(k-1)$ in Lemma~\ref{lem:VandA} for score $s'$ (define $V(0)=0$). Since $s < s'$ we have $V\leq V'$.

Now let $\tilde{c}_j\in \C$ be arbitrary with  $v(c_k) > v(\tilde{c}_j)$ and threshold $\tilde{t}_j$. We want to prove that $\tilde{c}_j$ cannot be chosen in $C_{s'}(v)$ for the $k$th position, which implies that $v(c^k)\leq v(d^k)$. If $\tilde{t}_j\geq t_k$ then $\tilde{c}_j$ is dominated: it will never be chosen. In particular, $\tilde{c}_j\notin C_{s'}(v_i)$. So suppose that $\tilde{t}_j<t_k$. Note that we can also rule out that $v(\tilde{c}_j) < V$ as $V\leq V'$ and $V'\leq v(d^k)$ (Lemma~\ref{lem:VandA}). So we know that $v(\tilde{c}_j)\geq V$. Then we obtain
\[
\begin{split}
F(s- t_k) v(c_k) + (1-F(s- t_k)) V 
& - [F(s-\tilde{t}_j) v(\tilde{c}_j) + (1-F(s- \tilde{t}_j)) V ] \\
& = F(s-t_k) [v(c_k)-V] - F(s- \tilde{t}_j) [v(\tilde{c}_j)-V] \\
& = F(s- t_k) v(c_k) - F(s- \tilde{t}_j) v(\tilde{c}_j) \\ & - (F(s- t_k)  - F(s- \tilde{t}_j)) V,    \\
\end{split}
\] two different representations for the difference in expected utility from applying to $c_k$ over applying to $\tilde{c}_j$ for the $k$th choice. We shall show that the first representation implies that the difference is increasing in $s$, and the second that it is increasing in $V$. 

Note that $s\mapsto F(s-t_k) - F(s-\tilde{t}_j)$ is strictly monotone increasing, as $F$ is concave and $t_k> \tilde{t}_j$. Since $[v(c_k)-V] > [v(\tilde{c}_j)-V]\geq 0$, $s\mapsto  F(s- t_k) [v(c_k)-V] - F(s- \tilde{t}_j) [v(\tilde{c}_j)-V] $ is also strictly monotone increasing. Note also that 
$- (F(s- t_k)  - F(s- \tilde{t}_j)) >  0$ so that $V\mapsto - (F(s- t_k)  - F(s- \tilde{t}_j)) V$ is strictly monotone increasing. 

The last step is a revealed preference argument. Given that $c_k\in C_s(v)$ is chosen at position $k$, we have:
\begin{align*}
    0 & \leq  F(s- t_k) v(c_k) + (1-F(s- t_k)) V 
 - [F(s- \tilde{t}_j) v(\tilde{c}_j) + (1-F(s- \tilde{t}_j)) V ] \\
 & <  F(s'- t_k) v(c_k) + (1-F(s'- t_k)) V 
 - [F(s'- \tilde{t}_j) v(\tilde{c}_j) + (1-F(s'- \tilde{t}_j)) V ] \\
  & \leq  F(s'- t_k) v(c_k) + (1-F(s'- t_k)) V' 
 - [F(s'- \tilde{t}_j) v(\tilde{c}_j) + (1-F(s'- \tilde{t}_j)) V' ].
\end{align*} The second inequality follows because the difference in expected utilities is strictly increasing in $s$ and $s<s'$. The third inequality because it is strictly increasing in $V$, and $V\leq V'$. 

Thus, if we have some $\tilde{c}_j$ for which $v(c_k) > v(\tilde{c}_j)$ and $\tilde{c}_j\in C_{s'}(v)$, then it is not chosen for the $k$th choice.
\end{proof}

Now we are in a position to prove Theorem~\ref{thm:main}.

\begin{proof} First we show that if  $s,s'\in S$ with $s< s'$ then for each $v$ there is a one-to-one function $f_v:C_s(v)\to C_{s'}(v)$ such that $c\leq f_v(c)$ for all $c\in C_s(v)$.

By Lemma~\ref{lem:paloapalo}, if $c\in C_s(v)$ and $s' > s$ then there exists $c' \in C_{s'}(v) $ such that $v(c)\leq v(c')$. Moreover $c'$ is chosen at the same position as $c$, since there are at least as many positions in $C_{s'}(v)$ as there are in $C_{s}(v)$ ($k^v_s\leq k^v_{s'}$). Observe that we cannot have $c'< c$ because such $c'$ would dominate $c$ (provide a higher utility at a higher probability of acceptance) and contradict that $c\in C_s(v)$.  Let $c'=f_v(c)$. Since different $c$ in $C_s(v)$ are chosen at different positions, $f_v$ is one-to-one. 

Now let $s< s'$ and note that $\bar G_{s'}(c) - \bar G_s(c)$ equals
\begin{align*} 
\sum_{c'\geq c} \int_{\Re^{\C}}  \left[\one_{c'\in C_{s'}(v)}  - \one_{c'\in C_s(v)} \right] d \la(v) 
& = \int_{\Re^{\C}}
\left[ \abs{\{c'\in C_{s'}(v) : c'\geq c \}}\right. \\ 
& \left. - \abs{\{c'\in C_{s}(v) : c'\geq c \}} \right] d \la(v) \\
& = \int_{\Re^{\C}} \abs{\{c'\in C_{s'}(v) \setminus f_{v} (C_{s}(v)) : c'\geq c \}} d \la(v) \\
& \geq 0.
\end{align*}
\end{proof}

\begin{remark}\label{rmk:strictcomparison}
Moreover, if $k^v_{s'}>k^v_s$ then the last inequality in the proof is strict as the $\{c'\in C_{s'}(v) \setminus f_{v} (C_{s}(v)) : c'\geq c \}\neq \os$ (by Lemma~\ref{lem:paloapalo}). In consequence, if $s\mapsto k_s$ is strictly increasing, then $\bar G_{s'}(c)> \bar G_s(c)$ when $s'>s$.
\end{remark}

\clearpage
\bibliographystyle{te}
\bibliography{colleges}

\clearpage
\appendix
\section{Appendices}

\ifnum\figmode=1
\subsection{Figures and Tables}

\begin{figure}[h]\caption{Grouped Test Score Distribution}
\label{fig:scoredist2}
\begin{center}
\includegraphics[width=0.95\columnwidth]{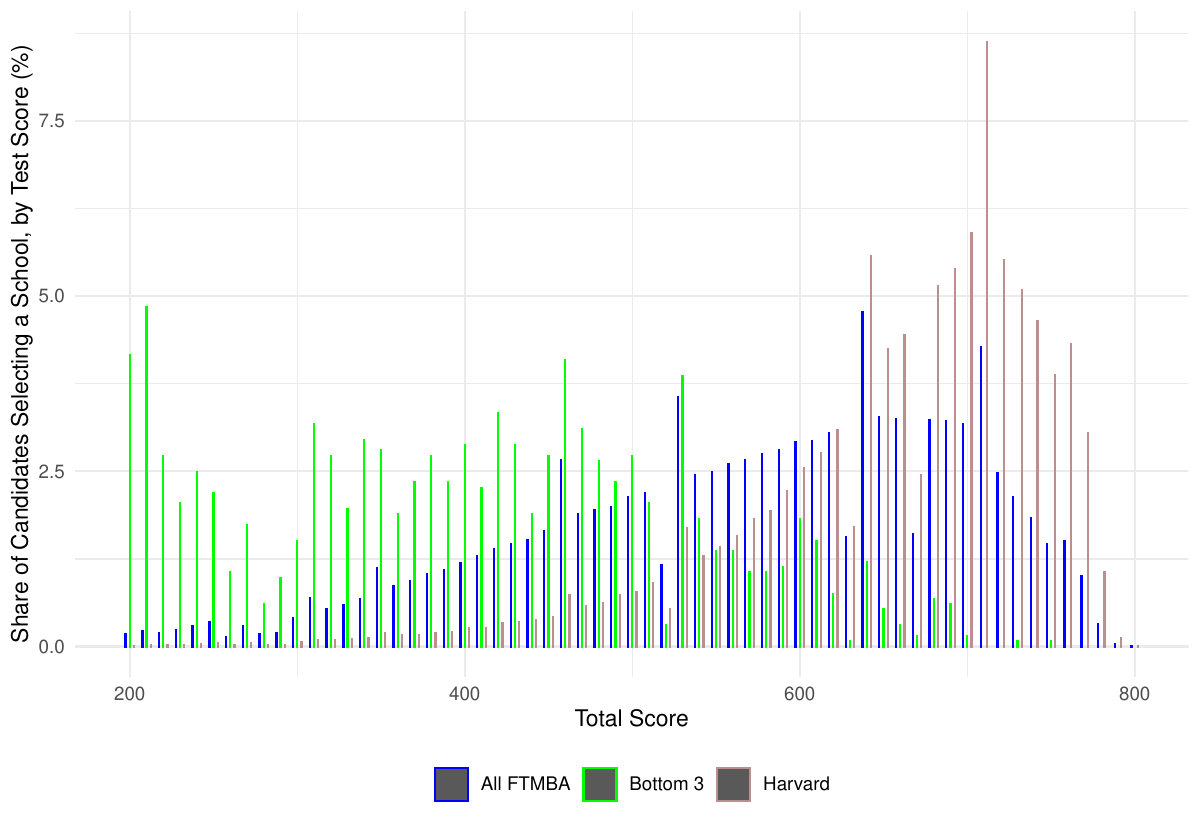}
\includegraphics[width=0.95\columnwidth]{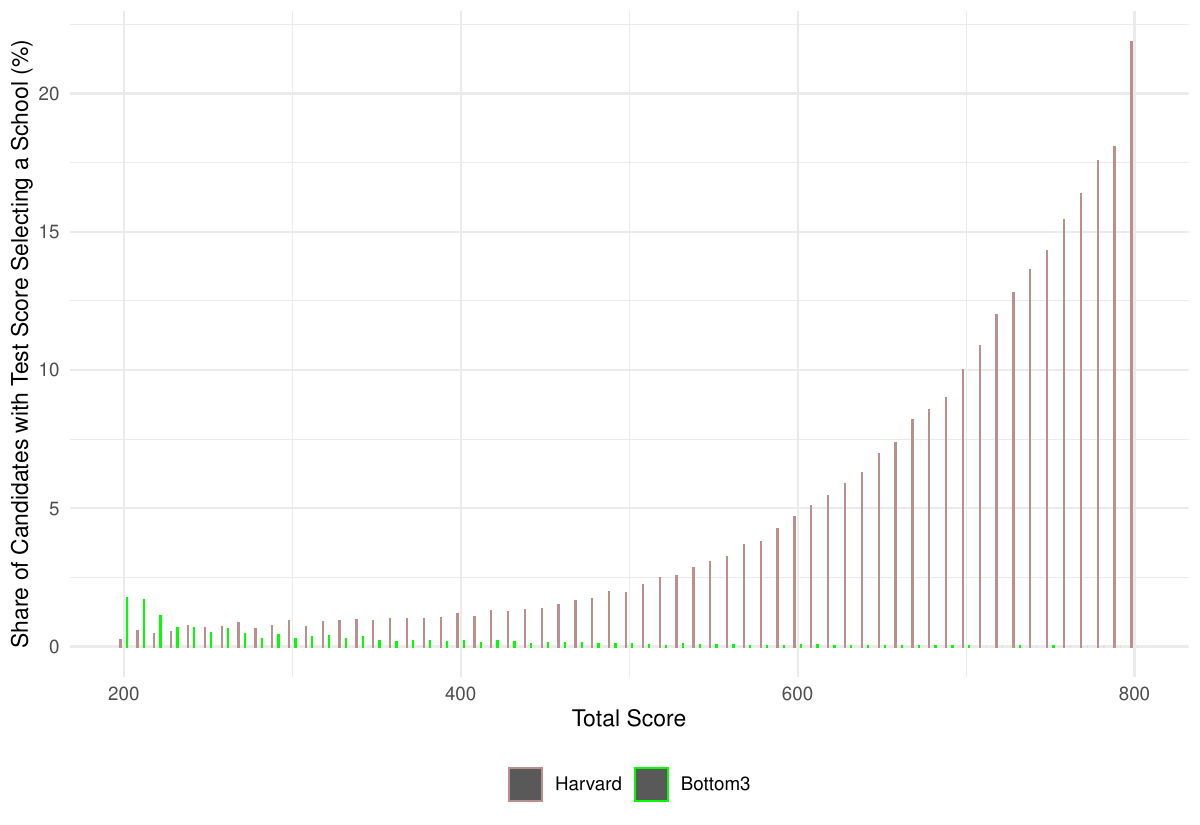}
\end{center}
\begin{minipage}[h]{.97\linewidth}\footnotesize{}{Note: The x-axis shows the range of GMAT Total Scores. The y-axis in the first panel shows the percentage of candidates in each group that received the test score on the x-axis. The first group represents all FTMBA program candidates, the second group is the subset of candidates that selected one of three low-ranked programs, the last group is the subset of candidates that selected Harvard, as one example of a top school. The second panel approximates a conditional probability distribution, showing the share of candidates with a given test score, that selected into one of the two sub-groups in panel 1.}\\
\footnotesize{}{Data Source: GMAC}
\end{minipage}
\end{figure}

\begin{landscape}
\begin{figure}[htbp]\caption{Heatmap of Candidate Total Scores and Program Average Scores}
\label{fig:heatmap}
\begin{center}
\includegraphics[width=0.9\columnwidth]{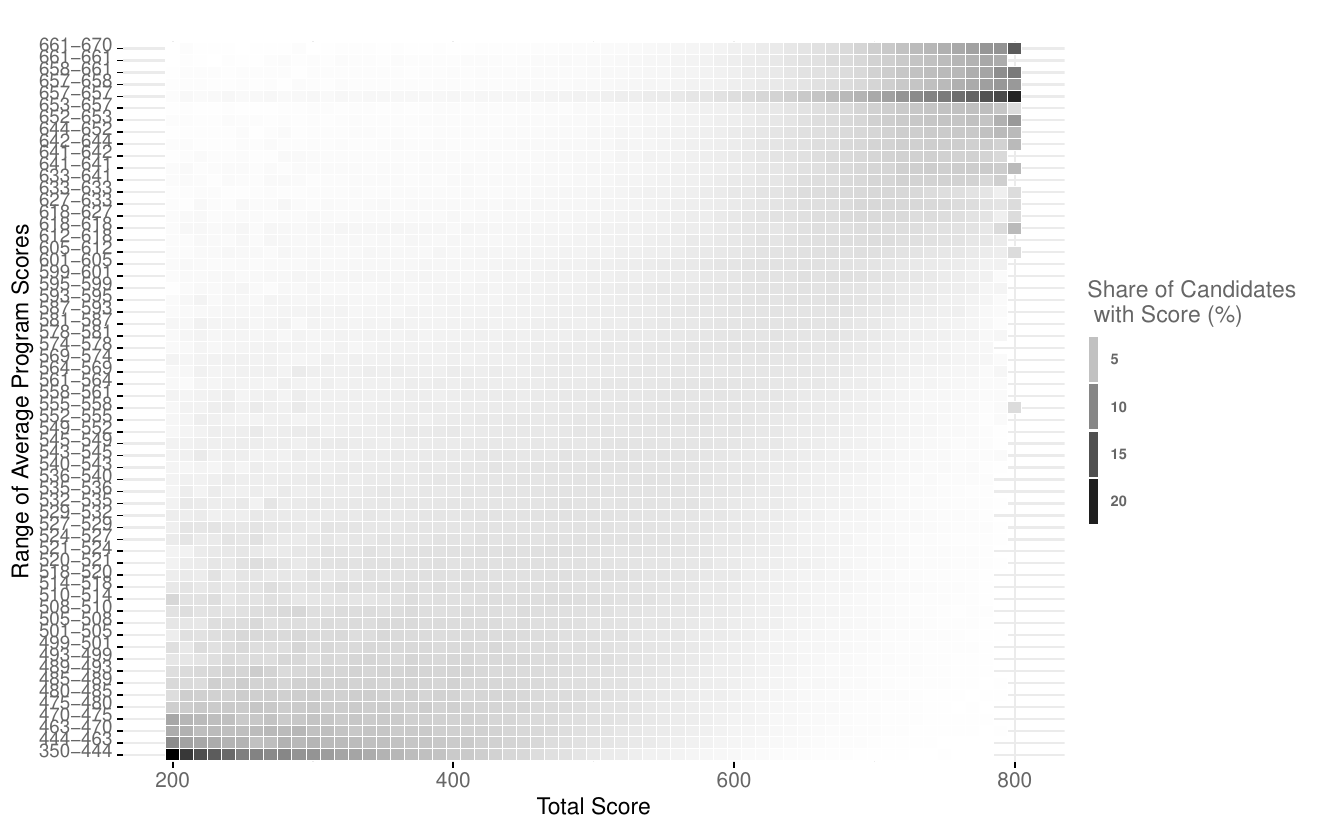}
\end{center}
\begin{minipage}[h]{.97\linewidth}
\footnotesize{}{Note: The graph shows the fraction of candidates with each test score that selected school program(s) in a score category. The score categories are defined by the average of GMAT Scores sent to the school programs. We set the number of categories to match the number of possible GMAT scores. Programs with fewer than 900 score reports are excluded.}\\
\footnotesize{}{Data Source: GMAC}
\end{minipage}
\end{figure}
\end{landscape}

\begin{figure}[h]
\begin{tikzpicture}
    \begin{axis}[
        axis lines=middle, 
        xmin=-0.05, xmax=3.5, 
        ymin=-0.05, ymax=3.5,  
        samples=100,
        domain=0:3,
        xlabel={\footnotesize $p(c)$},xlabel style={below},
        ylabel={\footnotesize $v(c)$},ylabel style={left},
        clip=false,
        xtick=\empty, ytick=\empty  
    ]
        \addplot[blue, thick, domain=0:sqrt(5)] {max(0, 2 - 0.4*x^2)}; 
        
        \pgfmathsetmacro\xone{0.6}
        \pgfmathsetmacro\xtwo{1.6}
        \pgfmathsetmacro\fone{2 - 0.4*\xone^2}
        \pgfmathsetmacro\ftwo{2 - 0.4*\xtwo^2}
        \pgfmathsetmacro\slopeone{-0.8*\xone}
        \pgfmathsetmacro\slopetwo{-0.8*\xtwo}
        
        \addplot[red, dashed, domain=0:1.4] {max(0, \fone + \slopeone*(x-\xone) + 0.2*(x-\xone)^2)} node[above right] {\footnotesize $p(c)v(c) + (1-p(c)) V'$};
        \addplot[violet, dashed, domain=1.0:2.4] {max(0, \ftwo + \slopetwo*(x-\xtwo) + 0.2*(x-\xtwo)^2)} node[above right] {\footnotesize $p(c)v(c) + (1-p(c)) V$};
        
        \addplot[black, only marks, mark=*] coordinates {(\xone,\fone) (\xtwo,\ftwo)};
    \end{axis}
\end{tikzpicture}
\caption{Choosing the $k$th school.}\label{fig:theory}
\end{figure}


\begin{table}[htbp] 
\begin{center}
  \caption{$M$ and $M+$ Measures} 
  \label{tab:t2} 
\resizebox{0.99\textwidth}{!}{
\begin{tabular}{@{\extracolsep{5pt}} clrr} 
\\[-1.8ex]\hline 
\hline \\[-1.8ex] 
 & University &  $m$ value & $m+$ value \\ 
\hline \\[-1.8ex] 
1 & Harvard University & $60,840.42$ & $3,612$ \\
2 & Stanford University & $51,376.91$ & $3,588$ \\
3 & University of Pennsylvania & $41,251.79$ & $3,560$ \\
4 & Columbia University & $31,595.05$ & $3,536$ \\
5 & Massachusetts Institute of Technology (MIT) & $29,497.54$ & $3,498$ \\
6 & Northwestern University & $27,297.17$ & $3,552$ \\
7 & University of Chicago & $23,775.61$ & $3,502$ \\
8 & University of California - Berkeley & $19,813.34$ & $3,492$ \\
9 & New York University & $17,603.50$ & $3,340$ \\
10 & Duke University & $11,388.95$ & $3,352$ \\
11 & Dartmouth College & $10,736.16$ & $3,518$ \\
12 & University of California - Los Angeles & $9,972.73$ & $3,122$ \\
13 & Yale University & $9,498.43$ & $3,258$ \\
14 & University of Virginia & $8,207.52$ & $3,190$ \\
15 & University of Michigan - Ann Arbor & $8,098.09$ & $3,332$ \\
16 & University of Texas - Austin & $5,615.18$ & $2,876$ \\
17 & Cornell University & $4,467.38$ & $3,114$ \\
18 & Georgetown University & $4,204.24$ & $2,948$ \\
19 & University of Southern California & $3,638.97$ & $2,708$ \\
20 & University of North Carolina - Chapel Hill & $3,523.65$ & $2,802$ \\
21 & Carnegie Mellon University & $2,870.80$ & $3,108$ \\
22 & Northwestern University$^*$ & $2,213.34$ & $3,020$ \\
23 & University of Notre Dame & $1,729.40$ & $2,702$ \\
24 & University of Washington - Seattle & $1,555.86$ & $2,614$ \\
25 & Vanderbilt University & $1,145.81$ & $2,432$ \\
26 & Boston College & $1,077.92$ & $2,468$ \\
27 & University of Minnesota - Duluth & $1,069.98$ & $1,406$ \\
28 & Indiana University - Bloomington & $992.51$ & $2,452$ \\
29 & Boston University & $984.37$ & $2,420$ \\
30 & Emory University & $855.64$ & $2,426$ \\
31 & Brigham Young University & $852.71$ & $2,334$ \\
32 & Washington University in St. Louis & $723.16$ & $2,676$ \\
33 & University of Minnesota - Twin Cities & $657.06$ & $2,392$ \\
34 & University of Wisconsin - Madison & $584.12$ & $2,366$ \\
35 & Babson College & $569.23$ & $2,348$ \\
\hline
\end{tabular} 
}   
\end{center} 
\raggedright{\footnotesize{}{The table shows our $m$-measure, based on the fraction of students with each test score that selected a school as part of the GMAT exam, $m(c) = \sum_{s,s'\in S}(g_{s'}(c) - g_s(c))(s'-s)$. We show the composite values of $m$ as well as the $mp(c) = \sum_{s,s'\in S}\left[(g_{s'}(c) - g_s(c))(s'-s)\right]>0$ for the top 35 programs. The table is ordered by the $m$ measure. $^*$ Programs within each school are ranked separately, this line shows the one-year accelerated FTMBA program for Northwestern University, not the regular 2-year FTMBA program on line 6.}}
\end{table}

\begin{table}[htbp] 
\begin{center}
      \caption{School Programs Tournament Rank}
  \label{tab:t5}
\resizebox{0.85\textwidth}{!}{
\begin{tabular}{@{\extracolsep{5pt}} clr}
\\[-1.8ex]\hline
\hline \\[-1.8ex]
Rank & University & wins \\
\hline \\[-1.8ex]
1 & Harvard University & $687$ \\
2 & Massachusetts Institute of Technology (MIT) & $686$ \\
3 & University of Pennsylvania & $685$ \\
4 & Stanford University & $684$ \\
5 & University of Chicago & $683$ \\
6 & Northwestern University & $682$ \\
7 & Dartmouth College & $681$ \\
8 & Columbia University & $680$ \\
9 & University of California - Berkeley & $679$ \\
10 & Yale University & $678$ \\
11 & Duke University & $677$ \\
12 & University of Virginia & $676$ \\
13 & University of Michigan - Ann Arbor & $675$ \\
14 & Cornell University & $674$ \\
15 & Carnegie Mellon University & $673$ \\
16 & Cornell University$*$ & $672$ \\
17 & New York University & $671$ \\
18 & Northwestern University$*$ & $670$ \\
19 & University of California - Los Angeles & $669$ \\
20 & Georgetown University & $668$ \\
21 & University of Texas - Austin & $667$ \\
22 & University of North Carolina - Chapel Hill & $666$ \\
23 & University of Southern California & $664$ \\
24 & Washington University in St. Louis & $663$ \\
25 & University of Notre Dame & $663$ \\
26 & Brandeis University & $662$ \\
27 & Emory University & $660$ \\
28 & University of Washington - Seattle & $659$ \\
29 & University of Michigan - Ann Arbor$**$ & $659$ \\
30 & Boston College & $657$ \\
31 & Indiana University - Bloomington & $656$ \\
32 & Boston University & $655$ \\
33 & Duke University$**$ & $655$ \\
34 & Vanderbilt University & $654$ \\
35 & University of Wisconsin - Madison & $653$ \\
\hline \\[-1.8ex]
\end{tabular}
}
\end{center}
\raggedright{\footnotesize{}{The table shows the top 35 FTMBA programs from a tournament ranking approach. All 688 MBA programs are matched against one another, with programs notching a point for each candidate that selected the program but did not select the other one, when at least one person with a higher score selected the other program. Having more points in each of the 687 match-ups counts as a win for a school program. Section \ref{sec:i2} explains the details of the match-ups.
$^*$ shows one-year accelerated FTMBA programs, (not the regular 2-year FTMBA), while $^{**}$ shows the Cross-Continent/Global MBA program for the University.}}
\end{table}

\begin{table}[htbp]
    \begin{center}
    \caption{Rankings by Undergrad Major Categories}
\resizebox{0.99\textwidth}{!}{
\begin{tabular}{@{\extracolsep{5pt}} clrrc}
\\[-1.8ex]\hline
\hline \\[-1.8ex]
 & University & Business-Major & Other & Rank 2 \\
\hline \\[-1.8ex]
1 & Harvard University & $54,895.85$ & $60,355.61$ & $1$ \\
2 & Stanford University & $51,829.62$ & $50,819.01$ & $2$ \\
3 & University of Pennsylvania & $41,695.77$ & $39,719.20$ & $3$ \\
4 & Columbia University & $30,825.70$ & $30,616.29$ & $5$ \\
5 & Massachusetts Institute of Technology (MIT) & $27,520.64$ & $32,092.05$ & $4$ \\
6 & Northwestern University & $26,933.38$ & $26,251.89$ & $6$ \\
7 & University of Chicago & $26,280.87$ & $21,219.54$ & $7$ \\
8 & New York University & $23,097.43$ & $16,029.62$ & $9$ \\
9 & University of California - Berkeley & $22,923.71$ & $20,524.71$ & $8$ \\
10 & Duke University & $10,441.33$ & $11,335.77$ & $10$ \\
11 & Dartmouth College & $10,259.33$ & $10,478.46$ & $11$ \\
12 & University of California - Los Angeles & $10,029.10$ & $9,312.07$ & $13$ \\
13 & University of Virginia & $9,498.72$ & $7,232.24$ & $15$ \\
14 & Northwestern University$^*$ & $8,789.46$ & $781.92$ & $28$ \\
15 & Yale University & $8,640.99$ & $9,791.30$ & $12$ \\
16 & University of Michigan - Ann Arbor & $6,556.73$ & $8,377.74$ & $14$ \\
17 & University of Texas - Austin & $5,630.72$ & $5,543.51$ & $16$ \\
18 & Georgetown University & $4,714.83$ & $3,394.81$ & $18$ \\
19 & University of North Carolina - Chapel Hill & $4,153.64$ & $2,725.33$ & $21$ \\
20 & Cornell University & $4,144.41$ & $4,675.49$ & $17$ \\
21 & University of Southern California & $3,974.11$ & $2,975.03$ & $20$ \\
22 & University of Notre Dame & $2,189.31$ & $1,310.22$ & $24$ \\
23 & Carnegie Mellon University & $2,018.70$ & $3,387.08$ & $19$ \\
24 & Vanderbilt University & $1,609.70$ & $638.18$ & $30$ \\
25 & University of Washington - Seattle & $1,382.98$ & $1,455.37$ & $22$ \\
26 & Indiana University - Bloomington & $1,238.08$ & $605.37$ & $31$ \\
27 & Brigham Young University & $1,236.05$ & $323.79$ & $36$ \\
28 & Boston College & $1,079.23$ & $953.53$ & $26$ \\
29 & Emory University & $1,076.24$ & $546.23$ & $33$ \\
30 & University of Wisconsin - Madison & $806.91$ & $328.44$ & $35$ \\
31 & University of Minnesota - Twin Cities & $741.67$ & $518.69$ & $34$ \\
32 & Washington University in St. Louis & $726.56$ & $729.35$ & $29$ \\
33 & Boston University & $553.75$ & $1,266.39$ & $25$ \\
34 & University of Rochester & $506.49$ & $230.83$ & $37$ \\
35 & Southern Methodist University & $381.45$ & $-229.04$ & $275$ \\
\hline \\[-1.8ex]
\end{tabular}
}
    \label{tab:maj}
    \end{center}
\raggedright{\footnotesize{}{
The table shows our $m$ measure, based on the schools students selected, calculated separately for broad groups of undergraduate majors. Business majors, the first column,  includes economics, and all business undergraduate degrees, while the `Other` column holds all other undergraduate majors. We show the composite values for the programs ranked as top 35. The table is ordered by the values of the Business majors group. The last column shows the rank for the category holding other undergrad majors. $^*$ This line refers to the one-year accelerated FTMBA program for the school. See Table \ref{tab:t2} notes.}}
\end{table}

\begin{table}[htbp]
    \begin{center}
    \caption{Rankings: Citizens and Internationals}
\resizebox{0.99\textwidth}{!}{
\begin{tabular}{@{\extracolsep{5pt}} clrrc}
\\[-1.8ex]\hline
\hline \\[-1.8ex]
 & University & US Citizens & Noncitizens & Rank 2 \\
\hline \\[-1.8ex]
1 & Harvard University & $60,580.74$ & $54,401.89$ & $1$ \\
2 & Stanford University & $51,069.98$ & $46,000.91$ & $3$ \\
3 & University of Pennsylvania & $39,431.28$ & $46,980.27$ & $2$ \\
4 & Columbia University & $29,992.23$ & $35,849.85$ & $4$ \\
5 & Massachusetts Institute of Technology (MIT) & $27,995.47$ & $32,292.72$ & $5$ \\
6 & Northwestern University & $27,360.08$ & $24,555.17$ & $7$ \\
7 & University of Chicago & $22,620.66$ & $27,787.66$ & $6$ \\
8 & University of California - Berkeley & $19,559.78$ & $19,939.65$ & $8$ \\
9 & New York University & $16,586.46$ & $18,910.66$ & $9$ \\
10 & Duke University & $11,409.93$ & $9,806.70$ & $11$ \\
11 & Dartmouth College & $11,084.70$ & $8,125.19$ & $13$ \\
12 & University of California - Los Angeles & $10,326.46$ & $8,480.00$ & $12$ \\
13 & Yale University & $9,323.67$ & $10,488.48$ & $10$ \\
14 & University of Virginia & $9,036.48$ & $5,036.51$ & $17$ \\
15 & University of Michigan - Ann Arbor & $7,900.16$ & $7,396.70$ & $14$ \\
16 & University of Texas - Austin & $5,992.27$ & $3,781.67$ & $18$ \\
17 & Georgetown University & $4,866.80$ & $1,687.01$ & $23$ \\
18 & Cornell University & $4,046.25$ & $6,530.12$ & $15$ \\
19 & University of Southern California & $3,823.35$ & $2,603.73$ & $20$ \\
20 & University of North Carolina - Chapel Hill & $3,428.01$ & $2,787.93$ & $19$ \\
21 & Carnegie Mellon University & $2,286.97$ & $5,670.38$ & $16$ \\
22 & University of Notre Dame & $1,980.46$ & $679.57$ & $30$ \\
23 & Northwestern University$^*$ & $1,932.38$ & $2,191.48$ & $21$ \\
24 & University of Washington - Seattle & $1,859.85$ & $686.09$ & $29$ \\
25 & Vanderbilt University & $1,366.04$ & $42.27$ & $51$ \\
26 & Boston University & $1,265.37$ & $477.63$ & $32$ \\
27 & Brigham Young University & $1,260.17$ & $-916.38$ & $555$ \\
28 & Boston College & $1,137.64$ & $862.89$ & $27$ \\
29 & University of Minnesota - Duluth & $1,113.18$ & $-136.19$ & $252$ \\
30 & Indiana University - Bloomington & $973.45$ & $1,300.68$ & $25$ \\
31 & Washington University in St. Louis & $877.13$ & $384.08$ & $35$ \\
32 & University of Minnesota - Twin Cities & $851.28$ & $164.90$ & $42$ \\
33 & University of Wisconsin - Madison & $700.24$ & $105.37$ & $46$ \\
34 & University of Colorado Boulder & $664.94$ & $-727.57$ & $526$ \\
35 & Babson College & $651.89$ & $330.02$ & $38$ \\

\hline \\[-1.8ex]
\end{tabular}
}
    \label{tab:tctzn}
    \end{center}
\raggedright{\footnotesize{}{
The table shows our $m$ measure calculated separately for US citizens and international students. The table is ordered by the values for citizens. The last column shows the rank for international students. $^*$ This line refers to the one-year accelerated FTMBA program for the school. See Table \ref{tab:t2} notes.}}
\end{table}

\begin{table}[htbp]
    \begin{center}
    \caption{Rankings: Early and Later Periods}
\resizebox{0.99\textwidth}{!}{
\begin{tabular}{@{\extracolsep{5pt}} clrrc}
\\[-1.8ex]\hline
\hline \\[-1.8ex]
 & University & Later & Early & Rank 2 \\
\hline \\[-1.8ex]
1 & Harvard University & $61,625.55$ & $62,316.47$ & $1$ \\
2 & Stanford University & $53,594.05$ & $50,939.35$ & $2$ \\
3 & University of Pennsylvania & $41,243.99$ & $41,514.80$ & $3$ \\
4 & Massachusetts Institute of Technology (MIT) & $33,887.49$ & $26,017.71$ & $5$ \\
5 & University of Chicago & $28,451.34$ & $20,297.96$ & $7$ \\
6 & Columbia University & $28,422.28$ & $35,225.33$ & $4$ \\
7 & Northwestern University & $27,817.83$ & $25,461.07$ & $6$ \\
8 & University of California - Berkeley & $20,807.98$ & $18,655.18$ & $9$ \\
9 & New York University & $15,737.72$ & $18,846.39$ & $8$ \\
10 & Duke University & $13,324.78$ & $9,575.18$ & $11$ \\
11 & Dartmouth College & $12,352.50$ & $9,118.52$ & $13$ \\
12 & Yale University & $9,570.94$ & $9,423.80$ & $12$ \\
13 & University of Michigan - Ann Arbor & $8,907.03$ & $6,913.51$ & $15$ \\
14 & University of California - Los Angeles & $8,547.90$ & $10,689.26$ & $10$ \\
15 & University of Virginia & $8,476.41$ & $7,938.16$ & $14$ \\
16 & University of Texas - Austin & $5,856.30$ & $5,402.97$ & $16$ \\
17 & Cornell University & $3,971.39$ & $4,795.05$ & $17$ \\
18 & Georgetown University & $3,581.05$ & $4,524.20$ & $18$ \\
19 & Northwestern University$^*$ & $3,322.52$ & $904.00$ & $29$ \\
20 & University of North Carolina - Chapel Hill & $3,222.41$ & $3,720.43$ & $20$ \\
21 & University of Southern California & $3,015.49$ & $3,984.26$ & $19$ \\
22 & Carnegie Mellon University & $2,738.56$ & $2,925.25$ & $21$ \\
23 & Vanderbilt University & $1,720.93$ & $842.10$ & $30$ \\
24 & University of Washington - Seattle & $1,632.21$ & $1,496.33$ & $24$ \\
25 & University of Notre Dame & $1,538.99$ & $1,820.17$ & $23$ \\
26 & Emory University & $1,280.38$ & $637.33$ & $34$ \\
27 & Indiana University - Bloomington & $1,032.00$ & $967.72$ & $28$ \\
28 & Boston College & $904.16$ & $1,219.73$ & $25$ \\
29 & University of Minnesota - Twin Cities & $775.52$ & $588.47$ & $35$ \\
30 & Boston University & $752.64$ & $1,177.05$ & $26$ \\
31 & Washington University in St. Louis & $737.05$ & $740.48$ & $32$ \\
32 & Brigham Young University & $645.98$ & $995.83$ & $27$ \\
33 & Cornell University$^**$ & $492.16$ & $567.24$ & $36$ \\
34 & University of Wisconsin - Madison & $380.31$ & $761.35$ & $31$ \\
35 & University of North Carolina - Chapel Hill & $363.10$ & $-34.46$ & $94$ \\
\hline \\[-1.8ex]
\end{tabular}
}
    \label{tab:prd}
    \end{center}
\raggedright{\footnotesize{}{
The table shows our $m$ measure  calculated separately for tests taken before July 1, 2011 (Early) and after (Later). The table is ordered by the values for more recent scores. The last column shows the rank for the early period. $^*$ This line refers to the one-year accelerated FTMBA program for the school. See Table \ref{tab:t2} notes.}}
\end{table}

\clearpage
\fi

\subsection{Additional Material}
The similarity between the $m$-value rankings and conventional school rankings can be seen in Figure~\ref{fig:ranks}. The pattern implies that using candidates' school choices effectively captures most of the information used to create expert rankings. 

\begin{figure}[htbp]\caption{US News Rank vs M-Values (Rank)}
\label{fig:ranks}
\begin{center}
\includegraphics[width=0.95\columnwidth]{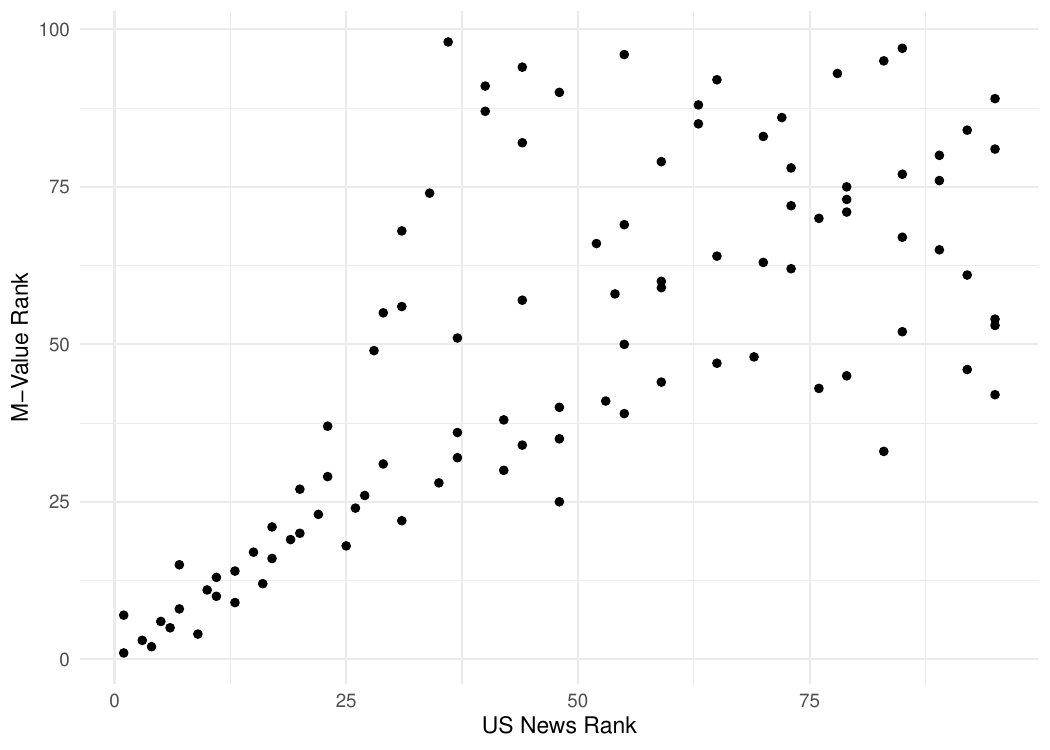}
\end{center}
\begin{minipage}[h]{.97\linewidth}\footnotesize{}{The figure shows the rank of the M-Value for each FTMBA program with a USNWR ranking.}\\
\footnotesize{}{Data Source: GMAC.}
\end{minipage}
\end{figure}

Table \ref{tab:tall} shows differences between the rankings of programs using more detailed groups of undergraduate majors. (The goal of the empirical exercise is a more detailed look at differences in selections by academic major, than the two broad groups in Table \ref{tab:maj}). Schools in the table are ordered by the ranking for Economics majors, and the columns of the table show changes in the rankings by undergraduate academic major (compared to Economics). A positive number is a drop in the rankings (or an increase away from 1), while a negative number is an improvement for the major in the ranking of the school program in the first column. 

It is necessary to clarify that the undergraduate majors in Table \ref{tab:maj} are groups created from the responses provided by test candidates. Candidates select their major from a drop-down menu that includes up to 42 categories, which, for convenience, are organized into the nine (9) used for this table. Candidates reported no undergraduate major, or chose ``My major is not listed'' for about 8.5\% of tests. The abbreviated names for degrees or groups of majors are: Accounting (Acct), Arts and Humanities (Arts), Engineering (Eng), Math and Computer Science (Math/CS), Other Social Sciences (SocSci), and Biz (Other business undergraduate degrees),

\begin{table}[!htbp] 
\begin{center}
  \caption{Rankings by Undergrad Major}
  \label{tab:tall}
\resizebox{0.99\textwidth}{!}{
\begin{tabular}{@{\extracolsep{5pt}} clrrrrrrrr}
\\[-1.8ex]\hline
\hline \\[-1.8ex]
 & University & Acct & Arts & Eng & Finance & Math/CS & Biz & SocSci & Sciences \\
\hline \\[-1.8ex]
1 & Harvard University & $0$ & $0$ & $0$ & $0$ & $0$ & $0$ & $0$ & $0$ \\
2 & Stanford University & $1$ & $0$ & $0$ & $1$ & $0$ & $0$ & $0$ & $0$ \\
3 & University of Pennsylvania & $-1$ & $1$ & $0$ & $-1$ & $0$ & $0$ & $0$ & $0$ \\
4 & Massachusetts Institute of Technology (MIT) & $5$ & $4$ & $0$ & $3$ & $0$ & $4$ & $3$ & $1$ \\
5 & Columbia University & $1$ & $-2$ & $2$ & $-1$ & $0$ & $0$ & $-1$ & $-1$ \\
6 & Northwestern University & $-1$ & $0$ & $-1$ & $0$ & $0$ & $-2$ & $-1$ & $0$ \\
7 & University of Chicago & $-3$ & $3$ & $-1$ & $-2$ & $0$ & $0$ & $2$ & $0$ \\
8 & University of California - Berkeley & $2$ & $-3$ & $0$ & $1$ & $0$ & $-2$ & $-2$ & $0$ \\
9 & New York University & $-1$ & $-2$ & $1$ & $-1$ & $0$ & $0$ & $-1$ & $0$ \\
10 & Dartmouth College & $3$ & $2$ & $1$ & $4$ & $3$ & $5$ & $0$ & $3$ \\
11 & University of California - Los Angeles & $-4$ & $0$ & $1$ & $1$ & $-1$ & $-1$ & $2$ & $1$ \\
12 & Yale University & $4$ & $-3$ & $3$ & $1$ & $2$ & $1$ & $0$ & $-1$ \\
13 & Duke University & $-2$ & $1$ & $-4$ & $-2$ & $-2$ & $-1$ & $-2$ & $-3$ \\
14 & Northwestern University$^*$ & $9$ & $32$ & $11$ & $10$ & $11$ & $12$ & $22$ & $17$ \\
15 & University of Virginia & $-3$ & $-2$ & $2$ & $-5$ & $0$ & $-4$ & $-1$ & $0$ \\
16 & University of Michigan - Ann Arbor & $-1$ & $-1$ & $-2$ & $-1$ & $-4$ & $-2$ & $0$ & $-2$ \\
17 & University of Southern California & $2$ & $1$ & $3$ & $6$ & $2$ & $0$ & $1$ & $1$ \\
18 & Georgetown University & $3$ & $-1$ & $3$ & $1$ & $2$ & $1$ & $-3$ & $2$ \\
19 & University of North Carolina - Chapel Hill & $-1$ & $3$ & $0$ & $-2$ & $3$ & $-1$ & $1$ & $0$ \\
20 & Cornell University & $-3$ & $0$ & $-4$ & $-2$ & $-3$ & $0$ & $-1$ & $-4$ \\
21 & University of Texas - Austin & $-7$ & $-5$ & $-8$ & $-5$ & $-5$ & $-5$ & $-4$ & $4$ \\
22 & Carnegie Mellon University & $4$ & $1$ & $-4$ & $0$ & $-4$ & $-1$ & $6$ & $-5$ \\
23 & University of Notre Dame & $-1$ & $5$ & $-1$ & $-2$ & $0$ & $-1$ & $-2$ & $0$ \\
24 & Emory University & $4$ & $6$ & $49$ & $5$ & $15$ & $4$ & $1$ & $10$ \\
25 & Boston College & $6$ & $-4$ & $569$ & $6$ & $5$ & $5$ & $1$ & $-1$ \\
26 & Vanderbilt University & $-2$ & $-1$ & $381$ & $-6$ & $173$ & $-3$ & $-2$ & $25$ \\
27 & University of Washington - Seattle & $0$ & $-3$ & $-3$ & $0$ & $0$ & $0$ & $-5$ & $-5$ \\
28 & University of Wisconsin - Madison & $2$ & $7$ & $421$ & $10$ & $9$ & $1$ & $10$ & $12$ \\
29 & University of Minnesota - Twin Cities & $4$ & $2$ & $4$ & $1$ & $-1$ & $3$ & $147$ & $-1$ \\
30 & Indiana University - Bloomington & $-5$ & $2$ & $-2$ & $-5$ & $-4$ & $-5$ & $3$ & $3$ \\
31 & University of Rochester & $12$ & $23$ & $5$ & $30$ & $14$ & $8$ & $9$ & $8$ \\
32 & University of California - Irvine & $2$ & $33$ & $476$ & $585$ & $1$ & $8$ & $226$ & $4$ \\
33 & Washington University in St. Louis & $7$ & $-4$ & $1$ & $-7$ & $-4$ & $1$ & $-6$ & $-4$ \\
34 & University of South Carolina - Columbia & $47$ & $5$ & $351$ & $21$ & $321$ & $17$ & $162$ & $411$ \\
35 & Northeastern University & $214$ & $14$ & $67$ & $223$ & $12$ & $30$ & $133$ & $448$ \\
\hline \\[-1.8ex]
\end{tabular}
}
\end{center}
\raggedright{\footnotesize{}{
The table shows rankings based on the $m$ measure for Economics majors, and differences between the ranking for Econ majors and the undergraduate majors in the columns. The abbreviated group names are Accounting (Acct), Arts and Humanities (Arts), Engineering (Eng), Math and Computer Science (Math/CS), Biz (Other business undergraduate degrees), and Other Social Sciences (SocSci). We show the programs ranked as top 35.  $^*$ This line refers to the one-year accelerated FTMBA program for the school. See Table \ref{tab:t2} notes.}}
\end{table}

Rankings are broadly similar across undergraduate majors, even when viewed at this level of detail. The overwhelming majority of changes in ranking across the columns in Table \ref{tab:maj}, are zero or in the single digits. Notably, the ranking of top 10 FTMBA schools based on the selections and scores of Economics majors, are almost identical to the ranking by undergrads from the sciences, as well as those from math and computer science. Engineering majors too resemble economics once one allows for the change from rank no. 5, where Columbia drops by two and other schools rise to take the spots. 

The notable exceptions to this pattern are interesting in how they allow tests of the intuition that some schools or programs have more of a reputation that appeal to candidates from certain undergraduate majors. The one-year accelerated MBA program at Northwestern, which candidates from economics seem to favor with a rank of 14, appears much less appealing to all other groups of undergraduate majors across the board. 

The biggest differences in rankings across undergraduate majors at this level of detail seem to be for Engineering, as well as other majors for schools ranked by economics majors near the bottom of the top 35. The movements in rank for engineering students by more than 380 (for Boston College, Vanderbilt, U-Wisconsin Madison and UC Irvine) represent strong shifts in preference across majors, an indication that other schools have reputations that are a bigger draw for students with engineering backgrounds. In general, as we go down the rankings, the movements across majors also increase. This is to be expected, given our proposition that schools offer a perceived generic value to candidates with less of an idiosyncratic component at the top of the ranking of schools.

Table \ref{tab:t4} shows the ordering of school programs by the total of the $m$ measure across all scores. Only the best test attempts for each student were included in the estimates for this table.

\begin{table}[!htbp] 
\begin{center}
  \caption{Total $M$ Measures for Best Attempts} 
  \label{tab:t4} 
\resizebox{0.99\textwidth}{!}{
\begin{tabular}{@{\extracolsep{5pt}} clrr} 
\\[-1.8ex]\hline 
\hline \\[-1.8ex]
 & University & $m$ value & $m+$ value \\ 
\hline \\[-1.8ex] 
1 & Harvard University & $58,560.70$ & $3,602$ \\
2 & Stanford University & $49,432.44$ & $3,540$ \\
3 & University of Pennsylvania & $39,110.48$ & $3,506$ \\
4 & Columbia University & $29,819.08$ & $3,490$ \\
5 & Massachusetts Institute of Technology (MIT) & $28,243.17$ & $3,486$ \\
6 & Northwestern University & $26,075.44$ & $3,524$ \\
7 & University of Chicago & $22,544.07$ & $3,438$ \\
8 & University of California - Berkeley & $18,835.79$ & $3,424$ \\
9 & New York University & $16,907.68$ & $3,372$ \\
10 & Duke University & $11,292.65$ & $3,368$ \\
11 & Dartmouth College & $10,067.02$ & $3,486$ \\
12 & University of California - Los Angeles & $9,708.73$ & $3,176$ \\
13 & Yale University & $9,067.91$ & $3,216$ \\
14 & University of Michigan - Ann Arbor & $8,022.84$ & $3,338$ \\
15 & University of Virginia & $7,814.38$ & $3,196$ \\
16 & University of Texas - Austin & $6,094.73$ & $2,938$ \\
17 & Georgetown University & $4,405.99$ & $2,994$ \\
18 & Cornell University & $4,332.02$ & $3,152$ \\
19 & University of Southern California & $3,923.27$ & $2,818$ \\
20 & University of North Carolina - Chapel Hill & $3,745.27$ & $2,828$ \\
21 & Carnegie Mellon University & $2,948.74$ & $3,110$ \\
22 & Northwestern University$^*$ & $2,202.13$ & $2,984$ \\
23 & University of Notre Dame & $2,029.64$ & $2,814$ \\
24 & University of Washington - Seattle & $1,880.88$ & $2,698$ \\
25 & Vanderbilt University & $1,424.44$ & $2,498$ \\
26 & Indiana University - Bloomington & $1,384.09$ & $2,596$ \\
27 & Boston University & $1,316.98$ & $2,530$ \\
28 & Boston College & $1,290.89$ & $2,572$ \\
29 & Brigham Young University & $1,183.83$ & $2,508$ \\
30 & University of Minnesota - Duluth & $1,071.71$ & $1,496$ \\
31 & Emory University & $1,018.46$ & $2,376$ \\
32 & Washington University in St. Louis & $942.98$ & $2,782$ \\
33 & University of Minnesota - Twin Cities & $849.95$ & $2,472$ \\
34 & University of Wisconsin - Madison & $787.22$ & $2,418$ \\
35 & Babson College & $750.07$ & $2,408$ \\
\hline \\[-1.8ex]
\end{tabular} 
}   
\end{center} 
\raggedright{\footnotesize{}{The table shows our measure based on the fraction of students with each test score that selected a school as part of the GMAT exam, $m(c) = \sum_{s\in S} \sum_{s'\in S}(g_{s'}(c) - g_s(c))(s'-s)$. The last column shows the sum of $m$ across all scores between 200 and 800. $^*$ This line refers to the one-year accelerated FTMBA program for the school. (See Table \ref{tab:t2} notes.)}}
\end{table} 

\subsection{Alternative empirical approaches}\label{sec:altempiricalstrategy}
Our empirical results are based on two different ways to operationalize the ideas expressed in Theorem~\ref{thm:main}. There are further alternatives that may prove fruitful in related problems, and we lay them out here. Specifically, we propose a class of heuristics for constructing the ranking $\leq$ from the data. Let $m_s(c)$ be the fraction of students with score $s$ who apply to college $c$.

So, $m_s(c)$ is the empirical density of college choices for students with score $s$.
Suppose that $\leq$ is the correct ranking over $\C$ and $f:\C\to \Re$ is a monotone increasing function. Then we know from Theorem~\ref{thm:main} that \[\text{if }  s< s' \then 
\sum_{c\in \C} f(c) m_s(c) - \sum_{c\in \C} f(c) m_{s'}(c) \leq 0,
\]  In other words, when we calculate the expectation of $f$ with respect to the empirical distribution of applications for each score, the expectation of the larger score is greater than the expectation of the lower score. 

Fix a family of functions $\mathcal{F}$, then we want to choose $f\in \mathcal{F}$ so as to minimize
 \[ \sum_{l=1}^{T-1}
\max\{ \sum_{c\in \C} f(c) m_{s^{l+1}}(c) - \sum_{c\in \C} f(c) m_{s^l}(c), 0 \}
\]
(recall that  $s^1> s^2 > \cdots > s^T$). When the  family of functions $\mathcal{F}$ is small, or has a conveniently succinct representation, we obtain a method for recovering the ordering $\geq$. 

In particular, imagine that we have a vector of observable characteristics $x_c\in \Re^n$ for each school $c\in \C$. Then we can let $\mathcal{F}$ be the collection of functions $x\mapsto \beta \cdot x_c$, for each $\beta\in K\subseteq\Re^n$. 

Then we would be solving the problem
\[\begin{array}{cc}
\min_{\beta\in \Re^n} & 
\sum_{l=1}^{T-1}
\max\{ \sum_{c\in \C} \beta\cdot x_c m_{s^{l+1}}(c) - \sum_{c\in \C} \beta \cdot x_c m_{s^l}(c), 0 \}  \\
s.t & \beta\in K, 
\end{array} 
\] for some convex and compact set $K$, for example a compact rectangle in $\Re^n$ (alternatively, one could add a lasso-type penalty for large $\beta$s). 

The resulting optimization problem is a convex program, and therefore computationally easy to solve. It may, however, be convenient to ``bin'' score values for those values that do not represent large numbers of  students.  Given a solution $\beta$, we conclude that $c\geq c'$ if $\beta \cdot x_c\geq \beta \cdot x_{c'}$. One advantage of this approach is that we would obtain an estimate $\beta \cdot x_c$ for the common value of college $c$, and the coefficients in $
\beta$ could have an interpretation.

\end{document}